\documentclass[journal]{IEEEtran}

\usepackage{amsmath}
\usepackage{amsfonts}
\usepackage{amssymb}
\usepackage{amsthm}
\usepackage{tabularx}
\usepackage{mathrsfs} 

\usepackage{url}
\usepackage{hyperref}
\newtheorem{theorem}{Theorem}
\newtheorem{definition}{Definition}

\newtheorem{corollary}{Corollary}

\usepackage{stfloats}
\usepackage{float}
\usepackage{graphicx}
\usepackage{cite}
\usepackage{xcolor}
\usepackage{subfigure}
\renewcommand{\vec}[1]{\mathbf{#1}}

\makeatletter
\def\blfootnote{\xdef\@thefnmark{}\@footnotetext}
\makeatother

\begin{document}
	
		\title{\huge{Performance Analysis over Correlated/Independent  Fisher-Snedecor $\mathcal{F}$ Fading Multiple Access Channels}} 
	\author{Farshad~Rostami~Ghadi\IEEEmembership{}~and  Wei-Ping Zhu\IEEEmembership{}}
	\maketitle
	\begin{abstract}
	In this paper, we investigate the impact of correlated fading on the performance of wireless multiple access channels (MAC) in the presence and absence of side information (SI) at transmitters, where the fading coefficients are modeled according to the Fisher-Snedecor $\mathcal{F}$ distribution.
	Specifically, we represent two scenarios: ($i$) clean MAC (i.e, without SI at transmitters), ($ii$) doubly dirty MAC (i.e., with the non-causally known SI at transmitters). For both system models, we derive the closed-form expressions of the outage probability (OP) as well as the average capacity (AC) under independent fading conditions.
	Besides, exploiting copula theory, we obtain the exact analytical expressions for the OP and the AC under positive dependence fading conditions in both considered models. Finally, the validity of the analytical results is illustrated numerically.
	\end{abstract}
	\maketitle
	\blfootnote{\noindent Copyright (c) 2015 IEEE. Personal use of this material is permitted. However, permission to use this material for any other purposes must be obtained from the IEEE by sending a request to pubs-permissions@ieee.org.} 
	 	\blfootnote{\noindent F.R. Ghadi is with Departmento de Ingenieria  de Comunicaciones, Universidad de M\'alaga - Campus de Excelencia Internacional Andaluc\'ia Tech., M\'alaga 29071, Spain  (e-mail: $\rm farshad@ic.uma.es$)}
	 	
	 	\blfootnote{\noindent  W.P. Zhu is with Department of Electrical and Computer Engineering, Concordia University, Montreal, QC H3G 1M8, Canada. (e-mail: $\rm weiping@ece.concordia.ca$).}
	
	\blfootnote{Digital Object Identifier 10.1109/XXX.2021.XXXXXXX}
	\vspace{-3mm}
	\section{Introduction}\label{introduction}
Multiple access channel (MAC) is a fundamental model for uplink communications in multi-user information theory, which has a significant role in designing emerging technologies such as intelligent reflecting surface (RIS) multi-user communication systems for the fifth-generation (5G) beyond wireless networks. For instance, given the importance of determining the capacity region in the performance analysis of multi-user communication systems from an information-theoretic viewpoint, the capacity characterization of MAC has been recently investigated in \cite{zhangb2021intelligent} for the more complex RIS-aided multi-user communication systems compared with single-user systems, by defining two RIS deployment strategies. On the other hand, modeling reliable communication between a transmitter and a receiver in the presence of side information (SI) about the channel state has always been one of the main problems in information theory. Multiple access communications techniques with the non-causally known SI at the transmitters can be of great interest to meet the reliability challenges for applications like connected robotics and autonomous systems in sixth-generation (6G) technology \cite{saad2019vision}, since such knowledge (e.g., either channel state information (CSI) or interference awareness) can be leveraged to intelligently encode their information. So, by considering such information at transmitters, the destructive effects of the interference can be reduced and reliable communication with higher rates can be achieved. The use of SI at the transmitter was first introduced by Shannon \cite{shannon1958channels} for the single-user point-to-point (P2P) communication systems and then studied by Jafar \cite{jafar2006capacity} in multi-user communication systems in order to determine a general capacity region of MAC. By exploiting a random binning technique, Philosof-Zamir \cite{philosof2008technical} extended Jafar's model and provided new achievable rate regions for the  MAC with non-causally known SI at the encoders. Later, they investigated a two-user Gaussian MAC with SI at both transmitters (i.e., doubly dirty MAC) for the strong interference regime, on which the achievable rate regions suffer from a bottleneck effect dominated by the weaker user compared with the case of a clean MAC (i.e., without interference) \cite{philosof2011lattice}.

In wireless fading channels, the channel coefficients, and hence, the corresponding channel signal-to-noise ratios (SNRs) are random variables (RVs) the probability distributions of which have been studied extensively in the literature \cite{tse2005fundamentals,atapattu2011mixture,selim2015modeling}. In performance analysis of these channels such as MAC, it is typically ignored the dependence structures of the fading coefficients for the sake of tractability \cite{biglieri2016impact}, while the channel coefficients observed by each user are practically correlated. Such correlation heavily depends on the proximity of transmitters, the presence or absence of scatters around the receiver, and the physical environment. Therefore, one of the main challenges in the performance analysis of these correlated channels is to generate the multivariate distributions representing the unknown joint statistics of different fading coefficients. In other words, the need for a simple statistical procedure to describe the arbitrary dependence structures between random fading coefficients is inevitable. In this regard, one flexible approach to describe the dependence structures is to exploit the copula theory \cite{nelsen2007introduction}, which has been recently used in the performance analysis of wireless communication systems \cite{gholizadeh2015capacity,ghadi2020copula,9464253,jorswieck2020copula,besser2020copula,ghadi2020copula1,besser2020bounds,ghadi2022capacity}. In \cite{gholizadeh2015capacity}, the capacity of multiple-input multiple-output (MIMO) P2P channel with correlated Nakagami-$m$ coefficients was studied. The authors in \cite{ghadi2020copula} derived a closed-form expression for the outage probability (OP) over correlated Rayleigh fading clean MAC, bringing out the constructive effect of a negative dependence between fading channels in the system performance. In contrast, the authors in \cite{9464253} represented that considering SI at transmitters can improve the efficiency of correlated MAC under the positive dependence structure in terms of the OP. General bounds of the ergodic achievable rate and the OP for dependent slow-fading clean MAC were provided in \cite{jorswieck2020copula} and \cite{besser2020copula}, respectively. Besides, assuming correlated fading coefficients, the authors in \cite{ghadi2020copula1} obtained closed-form expressions for the secrecy metrics in secure communications, while the general bounds of the secrecy OP for dependent fading channels were derived in \cite{besser2020bounds}. Only recently, the concept of copula was also applied in \cite{ghadi2022capacity} to express the impact of fading correlation on the performance of backscatter communications.

Another important challenge in analyzing wireless communications is to accurate modeling of the statistical characteristics of propagation environments. In this regard, the Fisher-Snedecor $\mathcal{F}$ distribution has been recently introduced in \cite{yoo2017fisher} to correctly model the combined effects of shadowing and multipath fading in device-to-device (D2D) communications. Given the experimental channel data obtained therein for D2D and wearable communication links, especially at 5.8 GHz, it was reported that the Fisher-Snedecor $\mathcal{F}$ provides a better fit in most cases compared with the Generalized-$\mathcal{K}$ fading model. For instance, the authors in \cite{yoo2017fisher} showed that the Fisher-Snedecor $\mathcal{F}$ model provides a better tail matching of the empirical cumulative density function (CDF) for composite fading compared with the Generalized-$\mathcal{K}$ model. Thus, regarding the fact that the tail of the empirical CDF indicates deep fading, the proposed Fisher-Snedecor $\mathcal{F}$ model is more appropriate for practical applications of fading channel modeling in wireless networks. Another advantage of Fisher-Snedecor $\mathcal{F}$ model is that its probability density function (PDF) includes only elementary functions with respect to the RV, and is as such expected to lead to more tractable analysis than the Generalized-$\mathcal{K}$ model. Besides, the Fisher-Snedecor $\mathcal{F}$ distribution can be reduced to some specific fading models such as Nakagami-$m$, Rayleigh, and one-sided Gaussian distributions. Therefore, due to its statistical tractability, the performance analysis of different communications systems under Fisher-Snedecor $\mathcal{F}$ fading has been recently studied  \cite{badarneh2018sum,kong2018physical,badarneh2018secrecy,yoo2019comprehensive,zhao2019ergodic,cheng2020bivariate,zhang2020dual,makarfi2020physical,badarneh2020achievable,han2021secrecy}. Importantly, 
the closed-form expressions of the PDF and CDF for the sum of independent Fisher-Snedecor $\mathcal{F}$ RVs in terms of the multivariate Fox's $H$-function were derived in \cite{du2020sum}, and the exact analytical expressions of the PDF and the CDF for both product and ratio product of independent Fisher-Snedecor $\mathcal{F}$ RVs were obtained in \cite{badarneh2020product}. In most previous works, either the independent structure is applied or the classic linear correlation with an asymptotic formulation is considered between the fading channels. However, with all the above considerations, several practical questions in multi-user communications systems remain unanswered to date: ($i$) What is the effect of \textit{correlated} Fisher-Snedecor $\mathcal{F}$ fading coefficients on the performance of multi-user communications? ($ii$) How does fading severity affect the performance multi-user communications such as MAC? To the best of our knowledge, there has been no previous work utilizing copula theory to investigate the effect of channel correlation on the performance of multi-user communications systems under correlated Fisher-Snedecor $\mathcal{F}$ distributions. Motivated by the aforesaid observations, in this paper, we combine copula theory with conventional statistical techniques to analyze the performance of the MAC with and without SI, on which any arbitrary dependence pattern can be considered.  The main advantage of this approach is that it allows us to consider arbitrary dependence structures that go beyond the linear dependence. Specifically, the main contributions of our work are summarized as follows:

\textbullet \,We provide the general formulations for the OP and the average capacity (AC) in both correlated/independent Fisher-Snedecor $\mathcal{F}$ fading clean and doubly MAC models, which hold for any \textit{arbitrary} choice of copulas.

\textbullet \,We exemplify how the OP and the AC performances can be characterized in the analytical expression for the \textit{Clayton} copula under the positive dependence structure.

\textbullet \,We analyze the impact fading correlation on the performance of clean and doubly dirty MAC models in terms of the OP and the AC, by changing the copula dependence parameter within the defined range. We also examine the effect of fading severity on the efficiency of OP and AC in the proposed model.

The rest of this paper is organized as follows. Section \ref{system-model} describes the system model considered in our work. The concept of copula and the analytical expressions of the OP for both correlated/independent clear and doubly dirty MAC models are provided in section \ref{outage-probability}. The AC analysis is represented in section \ref{sec-cap} for the studied models. In section \ref{num-results}, the efficiency of analytical results is illustrated numerically. Finally, the conclusions are drawn in section \ref{conclusion}.
	\section{System Model}\label{system-model}
	\subsection{Clean MAC}
	We consider a two-user wireless MAC with independent sources, where the transmitters $t_i$ send the inputs $X_i$, $i\in\{1,2\}$ reliably to a common receiver $r$, respectively (see Fig. \ref{system-mac}). The inputs $X_i$ sent by transmitters $t_i$ over the channels are subjected to the average power constraint as $\mathbb{E}[|X_i|^2]\leq P_i$, respectively. It is assumed that all users and the receiver are single antenna based, so, the corresponding channel output $Y$ at the receiver $r$ is defined as $Y=\sum_{i=1}^{2}h_iX_i+Z$,
	where $Z$ denotes the additive white Gaussian noise (AWGN) with zero mean and variance $N$ (i.e., $Z\sim\mathcal{N}(0,N)$) at the receiver $r$, and $h_i$ are the corresponding fading channel coefficients, which are modeled by the Fisher-Snedecor $\mathcal{F}$ distribution with fading parameters $(m_{i,s},m_i)$ so that $m_{i,s}$ and $m_i$ represent the amount of shadowing of the root-mean-square (rms) signal power and the fading severity parameter, respectively.
		In a two-user block fading clean MAC with the coherent receiver (fading coefficients $h_i$ are known at the receiver $r$), the instantaneous capacity region is determined as follows \cite{tse2005fundamentals}:
		\begin{align}\nonumber
			&R_i\leq\frac{1}{2}\log_2\left(1+\frac{P_i|h_i|^2}{N}\right)\\
			&R_1+R_2\leq\frac{1}{2}\log_2\left(1+\frac{P_1|h_1|^2+P_2|h_2|^2}{N}\right),
		\end{align}
		where $R_i$ are the desired transmission rates for transmitters $t_i$, respectively.\vspace{0ex}
				\begin{figure}[t]
	\centering
	\subfigure[Clean MAC]{%
		\includegraphics[width=0.45\textwidth]{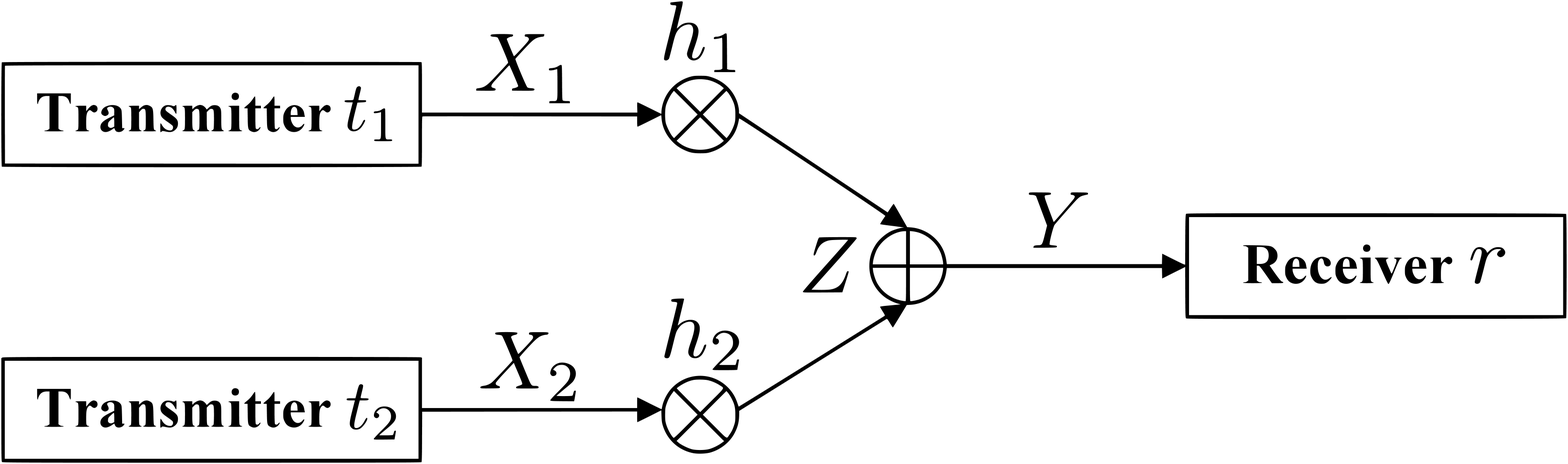}%
		\label{system-mac}%
	}\\\hspace{0cm}
	\subfigure[Doubly dirty MAC]{%
		\includegraphics[width=0.45\textwidth]{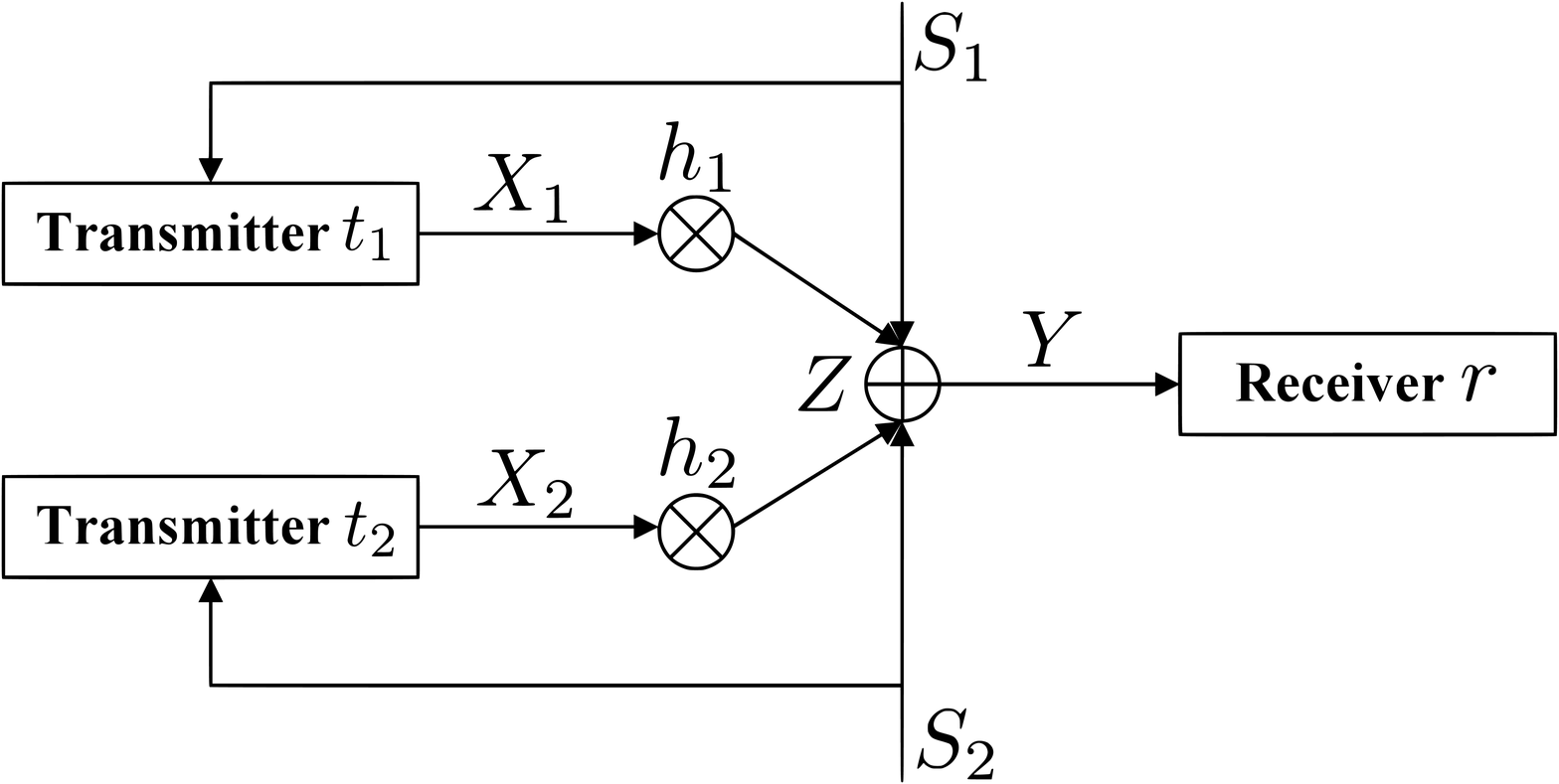}%
		\label{system-side}%
}\caption{System model depicting the multiple access communications: (a) Clean MAC; (b) Doubly dirty MAC.}\label{system}
	\end{figure}
	\subsection{Doubly dirty MAC}
	Here, we consider a two-user wireless doubly dirty MAC \cite{philosof2011lattice}, where two non-causally known interferences $S_i$ with variances $Q_i$ ($S_i\sim \mathcal{N}(0,Q_i)$) are added to transmitters $t_i$ in the clean MAC, respectively (see Fig. \ref{system-side}). Therefore, the corresponding channel output $Y$ at the receiver $r$ is defined as $Y=\sum_{i=1}^{2}\left(h_iX_i+S_i\right)+Z$.
		In a block fading doubly dirty MAC with the coherent receiver and two independent interferences $S_i$ non-causally known at transmitters $t_i$, the instantaneous capacity region is determined as follows as long as the interferences $S_i$ are strong (i.e., $Q_i \rightarrow\infty$)\cite{philosof2011lattice}
		\begin{align}
			R_1+R_2\leq \frac{1}{2}\log_2\left(1+\min\left\{\frac{P_1|h_1|^2}{N},\frac{P_2|h_2|^2}{ N}\right\}\right).\label{cap-si}
		\end{align}\vspace{-0.9cm}
	\subsection{SNR Distribution}
	For the concerned system models, the instantaneous SNR at the transmitters $t_i$ is defined as $\gamma_i=\frac{P_ig_i}{N}=\bar{\gamma}_ig_i$,
 where $g_i=|h_i|^2$ represents the instantaneous channel power gain with unit mean. The PDF and the CDF for the instantaneous SNRs $\gamma_i$ in terms of the Meijer’s G-function are respectively defined as \cite{yoo2017fisher,gradshteyn2014table}:
 \begin{align} f_i(\gamma_i)=\mathcal{A}_iG_{1,1}^{1,1}\left(
 \begin{array}{c}
 	\lambda_i\gamma_i\\
 \end{array}\Bigg\vert
 \begin{array}{c}
 	-m_{i,s}\\
 	m_i-1\\
 \end{array}\right),\label{pdf1}
\end{align}
\begin{align}
F_i(\gamma_i)=\mathcal{B}_iG_{2,2}^{1,2}\left(\begin{array}{c}
	\lambda_i\gamma_i\end{array}
\Bigg\vert\begin{array}{c}
	(1-m_{i,s},1)\\
	(m_i,0)\\
\end{array}\right),\label{cdf1}
\end{align}
%
	where $\lambda_i=\frac{m_i}{m_{i,s}\bar{\gamma}_i}$, $\mathcal{A}_{i}=\frac{\lambda_i}{\Gamma(m_i)\Gamma(m_{i,s})}$, and $\mathcal{B}_i=\frac{\gamma(m_i+1)}{m_i\Gamma(m_i)\Gamma(m_i,m_{i,s})\mathit{B}(m_i,m_{i,s})}$.
	\section{Outage Probability}\label{outage-probability}
	\subsection{Independent fading}
	In this section, we derive the analytical expressions of the OP under correlated/independent fading for both clean and doubly dirty MAC models.
	\subsubsection{Clean MAC}
	The  OP is  defined  as  the  probability  that  the  channel capacity is less than a certain information rate $R_t>0$. For the given clean MAC, we define the sum-rate OP as follows:
	\begin{align}
		P^{CM}_{out}&=\Pr\left(\frac{1}{2}\log_2\left(1+\gamma_1+\gamma_2\right)\leq R_t\right)\\
		&=\int_{0}^{\infty}F_1(\gamma_0)f_2({\gamma_2})d\gamma_2\label{app-out},
	\end{align}
	where $\gamma_0=2^{2R_t}-1-\gamma_2=\gamma_t-\gamma_2$ and $\gamma_t$ is the SNR threshold.
	\begin{theorem}\label{p-out-mac}
		The OP over independent Fisher-Snedecor $\mathcal{F}$ fading clean MAC is determined as \eqref{out-mac}.
		\begin{figure*}[!t]
			\small
				\setcounter{equation}{6}
			\begin{align}\label{out-mac}
				P^{CM}_{out}=-\gamma_t\mathcal{B}_1\mathcal{A}_2G_{1,0:2,3:1,1}^{0,1:2,1:1,1}\left(\begin{array}{c}
					\frac{-1}{\gamma_t\lambda_2},\frac{1}{\gamma_t\lambda_1}\end{array}
				\Bigg\vert\begin{array}{c}
					(2,1,1)\\
					-\\
				\end{array}\Bigg\vert
				\begin{array}{c}
					(1-m_1,1)\\
					(0,m_{1,s},1)\\
				\end{array}\Bigg\vert
				\begin{array}{c}
					(2-m_2)\\
					(1,1+m_{2,s})\\
				\end{array}\right).
			\end{align}\vspace{-0.5cm}
			\hrulefill
		\end{figure*}
	\end{theorem}
	\begin{proof}
		The details of the proof are in Appendix \ref{app-out-mac}.
	\end{proof}
	\subsubsection{Doubly dirty MAC}
	Similarly, the OP for the considered doubly dirty MAC can be defined as:
	\begin{align}
		P^{DM}_{out}&=\Pr\left(\frac{1}{2}\log_2\big(1+\min\left\{\gamma_1,\gamma_2\right\}\big)\leq R_t\right),\\
		&=F_1(\gamma_t)+F_2(\gamma_t)-F_1(\gamma_t)F_2(\gamma_t)\label{proof-side}.
	\end{align}
	\begin{theorem}\label{p-out-mac-side}
		The OP over independent Fisher-Snedecor $\mathcal{F}$ fading doubly dirty MAC is determined as \eqref{out-mac-side}.
		\begin{figure*}[!t]
			\small
				\setcounter{equation}{9}
			\begin{align}\nonumber
			 P^{DM}_{out}=\,&\mathcal{B}_1G_{2,2}^{1,2}\left(\begin{array}{c}
					\lambda_1\gamma_t\end{array}
				\Bigg\vert\begin{array}{c}
					(1-m_{1,s},1)\\
					(m_1,0)\\
				\end{array}\right)+\mathcal{B}_2G_{2,2}^{1,2}\left(\begin{array}{c}
					\lambda_2\gamma_t\end{array}
				\Bigg\vert\begin{array}{c}
					(1-m_{2,s},1)\\
					(m_2,0)\\
				\end{array}\right)\\
			&-\mathcal{B}_1\mathcal{B}_2G_{2,2}^{1,2}\left(\begin{array}{c}
				\lambda_1\gamma_t\end{array}
			\Bigg\vert\begin{array}{c}
				(1-m_{1,s},1)\\
				(m_1,0)\\
			\end{array}\right)G_{2,2}^{1,2}\left(\begin{array}{c}
				\lambda_2\gamma_t\end{array}
			\Bigg\vert\begin{array}{c}
				(1-m_{2,s},1)\\
				(m_2,0)\\
			\end{array}\right)\label{out-mac-side}.
			\end{align}
			\hrulefill\vspace{-0.7cm}
		\end{figure*}
	\end{theorem}
	\begin{proof}
		By substituting the  $F_i(\gamma_t)=\mathcal{B}_iG_{2,2}^{1,2}\left(\begin{array}{c}
			\lambda_i\gamma_t\end{array}
		\Bigg\vert\begin{array}{c}
			(1-m_{i,s},1)\\
			(m_i,0)\\
		\end{array}\right)$ for $i\in\{1,2\}$ into \eqref{proof-side}, the proof is completed. 
	\end{proof}\vspace{-0.5cm}
\subsection{Correlated fading}
First, we briefly review some basic definitions and properties of the two-dimensional copulas \cite{nelsen2007introduction}.
\begin{definition}[Copula]
	Let $\vec{V}=(V_1,V_2)$ be a vector of two RVs with marginal CDFs $F(v_j)=\Pr(V_j\leq v_j)$ for $j=1,2$, respectively. The relevant bivariate CDF is defined as:
	\begin{align}
			F(v_1,v_2)=\Pr(V_1\leq v_1,V_2\leq v_2).
	\end{align}
	Then, the copula function $C(u_1,u_2)$ of $\vec{V}=(V_1,V_2)$ defined on the unit hypercube $[0,1]^2$ with uniformly distributed RVs $U_j =F(v_j)$ for $j=1,2$ over $[0,1]$ is given by
	\begin{align}
		C(u_1,u_2)=\Pr(U_1\leq u_1,U_2\leq u_2).
	\end{align}
\end{definition}
\begin{theorem}[Sklar's theorem]\label{thm-sklar}
	Let $F(v_1,v_2)$ be a joint CDF of RVs with margins $F(v_j)$ for $j=1,2$. Then, there exists one copula function $C$ such that for all $v_j$ in the extended real line domain $\mathbb{R}$,
	\vspace{-1ex}
	\begin{align}\label{sklar}
		F(v_1,v_2)=C\left(F(v_1),F(v_2)\right).
	\end{align}
\end{theorem}
Since the distribution of the SNR for clean MAC is that of the sum of two correlated RVs, we exploit the following theorem to determine the CDF of $\gamma_s=\sum_{i=1}^{2}\gamma_i$. 
\begin{theorem}\label{thm-sum}
Let  $\textbf{V}=(V_1,V_2)$ be a vector of two absolutely continuous RVs with marginal CDFs $F_j(v_j)$, marginal PDFs $f_j(v_j)$ for $j=1,2$, and the corresponding copula $C$. Thus, the CDF of $\Omega=V_1+V_2$ is:
\begin{align}\label{eq-cdf}
F_\Omega(\omega)=\int_{\left[0,1\right]}C^{\prime}\left(u_1,F_2\left(\omega-F_1^{-1}(u_1)\right)\right)d\kappa,
\end{align}
where $\kappa$ generally denotes the Lebesgue measure, $C^\prime=\frac{\partial C}{\partial u_1}$ is the first partial derivative of a copula $C$, and $F^{-1}_{1}(.)$ is an inverse function of $F_{1}(.)$.\vspace{-1ex}
\end{theorem}
\begin{proof} By considering the convex set $\mathcal{G}_\omega=\{x_1+x_2\in\mathbb{R}^2:x_1+x_2\leq \omega\}$, $F_{\Omega}(\omega)$ can be formulated in terms of the copula measure $\mu_C$ as follows:
\begin{align}
&F_{\Omega}(\omega)=\Pr\left(V_1+V_2<\omega\right)=\int_{\mathcal{G}_\omega}f(v_1,v_2)dv_1dv_2,\\
&=\int_{\mathcal{G}_{\omega}}c\left(F_1(v_1),F_2(v_2)\right)f_1(v_1)f_2(v_2)dv_1dv_2,\\
&=\int_{\mathcal{H}_\omega}c\left(u_1,u_2\right)du_1du_2=\mu_C(\mathcal{H}_\omega)=\int_{\mathcal{H}_\omega}c\,d\kappa,\label{fubini}
\end{align}
where $c(u_1,u_2)$ denotes the copula density function, $f(v_1,v_2)$ is the joint PDF of RVs $V_j$,  $\mathcal{H}_\omega=\phi(\mathcal{G}_\omega)\subseteq[0,1]^2$, and  $\phi:\mathbb{R}^2\rightarrow[0,1]^2$. Denote further by $\mathcal{G}^*_{\omega}$ and $\mathcal{H}^*_{\omega}$ as the linear boundaries that respectively connected to $\mathcal{G}_{\omega}$ and $\mathcal{H}_{\omega}$, the set $\mathcal{H}^*_{\omega}=\phi(\mathcal{G}^*_{\omega})$ is given by $\mathcal{H}^*_{\omega}=\{(u_1,u_2)\in[0,1]^2,u_2=\tau_\omega(u_1)\}$, where $\tau_\omega(u_1)=F_2\left(\omega-F_1^{-1}(u_1)\right)$. Then, utilizing the Fubini theorem on \eqref{fubini}, the proof is completed. The details can be found in \cite{gijbels2014distribution} for a general case. 
\end{proof}
By taking the derivative with respect to $\omega$ in Theorem \ref{thm-sum},  the joint PDF for the sum of two correlated RVs can be obtained.
\begin{corollary}
By considering the assumptions of Theorem \ref{thm-sum}, the PDF of $\Omega=V_1+V_2$ is determined  as:
\begin{align}\nonumber
f_\Omega(\omega)=&\int_{\left[0,1\right]}c\left(u_1,F_2\left(\omega-F_1^{-1}(u_1)\right)\right)\\
&\times f_2\left(\omega-F_1^{-1}(u_1)\right)d\kappa.\label{pdf}
\end{align}
\end{corollary}
\begin{proof}
By utilizing definition $f_{\Omega}(\omega)=\frac{d}{d\omega}F_{\Omega}(\omega)$ and the allowance of interchanging integrals and
derivatives given the assumptions, the proof is completed. 
\end{proof}
\begin{definition}[Clayton copula]\label{Clayton}
		The bivariate Clayton copula with dependence parameter $\theta\in(0,\infty)$ is defined as:
		\begin{align}\label{def-cl}
		C_{\mathrm{cl}}\left(u_1,u_2\right)=\left({u_1}^{-\theta}+{u_2}^{-\theta}-1\right)^{-1/\theta},
		\end{align}
		which accepts 
		the positive dependence structure. When $\theta$ tends to 
		$+\infty$, the 
		upper bound of  Fr\'echet-Hoeffding \cite[Thm. 2.2.3]{nelsen2007introduction} will be attained (i.e., the perfect positive correlation is reached). Also, independence is achieved as $\theta$ approaches zero.
\end{definition}
\subsubsection{Clean MAC} The OP for a clean MAC under correlated fading condition can be defined as:
	\begin{align}
	P^{CM}_{out}&=\Pr\left(\frac{1}{2}\log_2\left(1+\gamma_1+\gamma_2\right)\leq R_t\right)=F_{\gamma_s}(\gamma_t),
\end{align}
where $\gamma_s=\gamma_1+\gamma_2$ and $\gamma_t=2^{2R_t}-1$.
\begin{theorem}
	The OP over correlated Fisher-Snedecor $\mathcal{F}$ fading clean MAC and any arbitrary copula $C$ is determined as: \vspace{-0.1cm}
	\begin{align}\label{thm-out-mac}
		P_{out}^{CM}=\int_{0}^{1} C^\prime\left(u_1,\tau_{\gamma_t}(u_1)\right)du_1,
	\end{align}\vspace{-0.2cm}
where, 
\begin{align}\nonumber
&\tau_{\gamma_t}(u_1)=\\ 
&I_{1-\frac{1}{\frac{m_2}{m_{2,s}} \left(\gamma_t-\frac{m_{1,s}}{m_1}\left(\frac{1}{I_{(1,-u_1)}^{-1}(m_{1,s},m_1)}-1\right)\right)+m_{2,s}}}(m_2,m_{2,s}), \label{tau}
\end{align}
and $I^{-1}_{(.,.)}(.,.)$ is the inverse of generalized regularized Beta function $I_{(.,.)}(.,.)$.
\end{theorem}
\begin{proof} To prove \eqref{thm-out-mac}, we need to generate the function under the integral in \eqref{eq-cdf}. Thus, we first rewrite the CDF of SNRs in terms of the regularized Beta function as follows: 
	\begin{align}\label{cdf-new}
F_{i}(\gamma_i)=I_{\frac{m_i \gamma_i}{m_i \gamma_i+m_{i,s}}}(m_i,m_{i,s}),
\end{align}
where $F_1^{-1}(u_1)=\frac{m_{1,s}}{m_1} \left(\frac{1}{I_{(1,-u_1)}^{-1}(m_{1,s},m_1)}-1\right)$, and by doing so, 
\begin{align}\label{inv-2}
F_{2}\left(\gamma_t-F_1^{-1}(u_1)\right)=I_{1-\frac{m_{2,s}}{m_2\left(\gamma_t-F_1^{-1}(u_1)\right)+m_{2,s}}}(m_2,m_{2,s}).
\end{align}
Now, by inserting $F_1^{-1}(u_1)$ into \eqref{inv-2} 
and exploiting Theorem \ref{thm-sum}, the proof is completed.
\end{proof}
Here, we use the Clayton copula to analyze the performance of the OP. This choice is justified because it captures the positive dependence between RVs for any range of correlation and covers the tail dependence, 
	while offering good mathematical tractability.
\begin{corollary}
The OP over correlated Fisher-Snedecor $\mathcal{F}$ fading clean MAC, using the Clayton copula is given by
\begin{align}
	P_{out}^{CM}=\int_{0}^{1}
	-u_1^{{\theta}-1}\left(u_1^{\theta}+\tau_{\gamma_t}(u_1)^{\theta}-1\right)^{-\frac{1}{{\theta}}-1}du_1.\label{col-out-mac}
\end{align}
where $\tau_{\gamma_t}(u_1)$ is defined in \eqref{tau}.
\end{corollary}
\begin{proof}
	By computing the first partial derivative for the Clayton copula $C_{\mathrm{cl}}$ as follows:
\begin{align}
C_{\mathrm{cl}}^\prime(u_1,u_2)=\frac{\partial C_{\mathrm{cl}}}{\partial u_1}=-u_1^{{\theta}-1}\left(u_1^{\theta}+u_2^{\theta}-1\right)^{-\frac{1}{{\theta}}-1},\label{der-Clayton}
\end{align}
and then, considering $\tau_{\gamma_t}(u_1)$ instead of $u_2$ in \eqref{der-Clayton} and inserting \eqref{der-Clayton} into \eqref{thm-out-mac}, the proof is completed. Besides, by considering the marginal CDF in \eqref{cdf-new} for $u_1$, \eqref{col-out-mac} can be evaluated numerically.
\end{proof}
\subsubsection{Doubly dirty MAC} The OP for a doubly dirty MAC under correlated fading condition can be defined as:
	\begin{align}
	P^{DM}_{out}&=\Pr\left(\frac{1}{2}\log_2\big(1+\min\left\{\gamma_1,\gamma_2\right\}\big)\leq R_t\right),\\
	&=F_1(\gamma_t)+F_2(\gamma_t)-F\left(\gamma_t,\gamma_t\right),\label{eq-ddm}
\end{align}
where $\gamma_t=2^{2R_t}-1$ and $F(\gamma_t,\gamma_t)$ can be obtained form Theorem \ref{thm-sklar} for any arbitrary copula $C$.
\begin{theorem}
	The OP over correlated Fisher-Snedecor $\mathcal{F}$ fading doubly dirty MAC and any arbitrary copula $C$ is determined as:
	\begin{align}\label{out-dd}
	P_{out}^{DM}=F_1(\gamma_t)+F_2(\gamma_t)-C\left(F_1\left(\gamma_t\right),F_2\left(\gamma_t\right)\right).
	\end{align}
\end{theorem}
\begin{proof}
By applying Theorem \ref{thm-sklar} to \eqref{eq-ddm}, the proof is completed.  
\end{proof}
\begin{corollary}
The OP over correlated Fisher-Snedecor $\mathcal{F}$ fading doubly dirty MAC, using the Clayton copula is determined as \eqref{out-ddm-corr}.
		\begin{figure*}[!t]
	\small
	\setcounter{equation}{29}
	\begin{align}\nonumber
		P^{DM}_{out}=\,&\mathcal{B}_1G_{2,2}^{1,2}\left(\begin{array}{c}
			\lambda_1\gamma_t\end{array}
		\Bigg\vert\begin{array}{c}
			(1-m_{1,s},1)\\
			(m_1,0)\\
		\end{array}\right)+\mathcal{B}_2G_{2,2}^{1,2}\left(\begin{array}{c}
			\lambda_2\gamma_t\end{array}
		\Bigg\vert\begin{array}{c}
			(1-m_{2,s},1)\\
			(m_2,0)\\
		\end{array}\right)\\
		&-\left(\left(\mathcal{B}_1G_{2,2}^{1,2}\left(\begin{array}{c}
			\lambda_1\gamma_t\end{array}
		\Bigg\vert\begin{array}{c}
			(1-m_{1,s},1)\\
			(m_1,0)\\
		\end{array}\right)\right)^{-\theta}+\left(\mathcal{B}_2G_{2,2}^{1,2}\left(\begin{array}{c}
			\lambda_2\gamma_t\end{array}
		\Bigg\vert\begin{array}{c}
			(1-m_{2,s},1)\\
			(m_2,0)\\
		\end{array}\right)\right)^{-\theta}-1\right)^{-1/\theta}.\label{out-ddm-corr}
	\end{align}
	\hrulefill\vspace{0cm}
\end{figure*} 
\end{corollary}
\begin{proof}
	By considering $u_1=F_1(\gamma_t)=\mathcal{B}_1G_{2,2}^{1,2}\left(\begin{array}{c}
		\lambda_1\gamma_t\end{array}
	\Bigg\vert\begin{array}{c}
		(1-m_{1,s},1)\\
		(m_1,0)\\
	\end{array}\right)$ and $u_2=F_2(\gamma_t)=\mathcal{B}_2G_{2,2}^{1,2}\left(\begin{array}{c}
	\lambda_1\gamma_t\end{array}
\Bigg\vert\begin{array}{c}
(1-m_{2,s},1)\\
(m_2,0)\\
\end{array}\right)$ in the Clayton copula definition \eqref{def-cl}, and then substituting the result into \eqref{out-dd}, the proof is completed. \vspace{-0.3cm}
	\end{proof}
\section{Average Capacity}\label{sec-cap}
In this section, we represent the analytical expressions of the AC under correlated/independent fading for both clean and doubly dirty MAC models. 
	\subsection{Independent fading}
\subsubsection{Clean MAC} Assuming $\gamma_s=\gamma_1+\gamma_2$, the AC for the given fading clean MAC is defined as:
\begin{align}
\bar{\mathcal{C}}^{CM}=\int_{0}^{\infty}\frac{1}{2}\log_2\left(1+\gamma_s\right)f_s(\gamma_s)d\gamma_s,\label{asc-def}
\end{align}
where $f_s(\gamma_s)$ denotes the PDF of $\gamma_s$. Given the independence of the SNRs $\gamma_1$ and $\gamma_2$, the PDF of $\gamma_s$ can be defined as:
\begin{align}
	&f_s(\gamma_s)=\int_{0}^{\infty}f_1(\gamma_1)f_{2}(\gamma_s-\gamma_1)d\gamma_1.\label{pdf-def}
	\end{align}
By inserting \eqref{pdf1} into \eqref{pdf-def} and exploiting \cite[eq. (2.24.1.3)]{prudnikov1990more}, $f_s(\gamma_s)$ is determined as:
\begin{align}\nonumber
f_s(\gamma_s)=&\mathcal{A}_1\mathcal{A}_2\sum_{k=0}^{\infty}\frac{(-\lambda_2)^{k+1}\gamma_s^k}{k!}\\
&\hspace{-0.1cm}\times G_{3,3}^{2,3}\left(
\begin{array}{c}
	\frac{-\lambda_1}{\lambda_2}\\
\end{array}\Bigg\vert
\begin{array}{c}
	(0,-m_{1,s},k-m_2+1)\\
	(m_1-1,k+m_{2,s},k)\\
\end{array}\right).\label{pdf-gs}
\end{align}
Now, by using the PDF obtained in \eqref{pdf-gs}, the AC for independent fading clean MAC is determined as the following theorem. 
\begin{theorem}\label{thm-ac-mac}
		The AC over independent Fisher-Snedecor $\mathcal{F}$ fading clean MAC is determined as
\begin{align}\nonumber
\bar{\mathcal{C}}^{CM}=&\frac{\mathcal{A}_1\mathcal{A}_2}{\ln 2}\sum_{k=0}^{\infty}\frac{(-\lambda_2)^{k+1}}{k!}\left[\frac{\Gamma\left(-\left(k+1\right)\right)}{\Gamma\left(-k\right)}\right]^2\\
&\times G_{3,3}^{2,3}\left(
\begin{array}{c}
	\frac{-\lambda_1}{\lambda_2}\\
\end{array}\Bigg\vert
\begin{array}{c}
	(0,-m_{1,s},k-m_2+1)\\
	(m_1-1,k+m_{2,s},k)\\
\end{array}\right).
		\end{align}
\end{theorem}
\begin{proof}
The details of the proof are in Appendix \ref{ac-mac}.
\end{proof}
\subsubsection{Doubly dirty MAC} Assuming $\gamma_n=\min\{\gamma_1,\gamma_2\}$, the AC for the considered fading doubly dirty MAC is defined as:
\begin{align}
\bar{\mathcal{C}}^{DM}=\int_{0}^{\infty}\frac{1}{2}\log_2\left(1+\gamma_n\right)f_n(\gamma_n)d\gamma_n,\label{def-ac-dd}
\end{align}
where $f_n(\gamma_n)$ denotes the PDF of $\gamma_n$ and the CDF of $\gamma_n$ is given by $F_n(\gamma_n)=\Pr\left(\min\{\gamma_1,\gamma_2\}\leq\gamma_n\right)$. Thus, by definition, the PDF of $\gamma_n$ is defined as:
\begin{align}
&\hspace{-0.3cm}f_n(\gamma_n)=f_1(\gamma_n)\left(1-F_2\left(\gamma_n\right)\right)+f_2(\gamma_n)\left(1-F_1\left(\gamma_n\right)\right).\label{pdf-gm}
\end{align}
Now, by inserting \eqref{pdf1} and \eqref{cdf1} into \eqref{pdf-gm}, and utilizing the respective results, the AC for the fading doubly dirty MAC model is determined as the following theorem. 
\begin{theorem}\label{thm-ac-dd}
The AC over independent Fisher-Snedecor $\mathcal{F}$ fading doubly dirty MAC is determined as \eqref{ac-dd}.
		\begin{figure*}[!t]
	\small
	\setcounter{equation}{36}
	\begin{align}\nonumber
	\bar{\mathcal{C}}^{DM}=&\frac{\mathcal{A}_1}{2\lambda_1\ln2}G_{3,3}^{2,3}\left(
	\begin{array}{c}
		\frac{1}{\lambda_1}\\
	\end{array}\Bigg\vert
	\begin{array}{c}
		(1,1,1-m_{1})\\
		(1,m_{1,s},0)\\
	\end{array}\right)-\frac{\mathcal{A}_1\mathcal{B}_2}{2\lambda_1\ln 2}G_{1,1:2,2:2,2}^{1,1:1,2:1,2}\left(\begin{array}{c}
	\frac{\lambda_2}{\lambda_1},\frac{1}{\lambda_1}\end{array}
\Bigg\vert\begin{array}{c}
(m_1)\\
(m_{1,s})\\
\end{array}\Bigg\vert
\begin{array}{c}
(1-m_{2,s},1)\\
(m_2,0)\\
\end{array}\Bigg\vert
\begin{array}{c}
(1,1)\\
(1,0)\\
\end{array}\right)\\
&+\frac{\mathcal{A}_2}{2\lambda_2\ln2}G_{3,3}^{2,3}\left(
\begin{array}{c}
	\frac{1}{\lambda_2}\\
\end{array}\Bigg\vert
\begin{array}{c}
	(1,1,1-m_{2})\\
	(1,m_{2,s},0)\\
\end{array}\right)-\frac{\mathcal{B}_1\mathcal{A}_2}{2\lambda_2\ln 2}G_{1,1:2,2:2,2}^{1,1:1,2:1,2}\left(\begin{array}{c}
\frac{\lambda_1}{\lambda_2},\frac{1}{\lambda_2}\end{array}
\Bigg\vert\begin{array}{c}
(m_2)\\
(m_{2,s})\\
\end{array}\Bigg\vert
\begin{array}{c}
(1-m_{1,s},1)\\
(m_1,0)\\
\end{array}\Bigg\vert
\begin{array}{c}
(1,1)\\
(1,0)\\
\end{array}\right).\label{ac-dd}
	\end{align}
	\hrulefill\vspace{0cm}
\end{figure*} 
\end{theorem}
\begin{proof}
By plugging \eqref{pdf-gm} into \eqref{def-ac-dd}, $\bar{\mathcal{C}}^{DM}$ can be rewritten as:
\begin{align}
	\bar{\mathcal{C}}^{DM}=\mathcal{J}_1-\mathcal{J}_2+\mathcal{J}_3-\mathcal{J}_4,\label{j1-j4}
\end{align}
where
\begin{align}
\mathcal{J}_1=\int_{0}^{\infty}\frac{1}{2}\log_2(1+\gamma_n)f_1(\gamma_n)d\gamma_n,\label{j1}
\end{align}
\begin{align}
	\mathcal{J}_2=\int_{0}^{\infty}\frac{1}{2}\log_2(1+\gamma_n)f_1(\gamma_n)F_2(\gamma_n)d\gamma_n,\label{j2}
\end{align}
\begin{align}
	\mathcal{J}_3=\int_{0}^{\infty}\frac{1}{2}\log_2(1+\gamma_n)f_2(\gamma_n)d\gamma_n,
\end{align}
\begin{align}
	\mathcal{J}_4=\int_{0}^{\infty}\frac{1}{2}\log_2(1+\gamma_n)f_2(\gamma_n)F_1(\gamma_n)d\gamma_n.
\end{align}
Then, by computing the integrals $\mathcal{J}_l$ for $l\in\{1,2,3,4\}$, the proof is completed. The details of the proof are in Appendix \ref{app-ac-dd}.
\end{proof}
	\subsection{Correlated fading}
\subsubsection{Clean MAC} By exploiting the copula theory and assuming the dependence between SNRs $\gamma_1$ and $\gamma_2$, the AC for the considered correlated fading clean MAC is derived as the following theorem. 
\begin{theorem}
The AC over correlated Fisher-Snedecor $\mathcal{F}$ fading clean MAC and any arbitrary copula density function $c$ is determined as:
\begin{align}\nonumber
	\bar{\mathcal{C}}^{CM}
	&=\int_{0}^{\infty}\int_{0}^{1}\frac{1}{2}\log_2\left(1+\gamma_s\right)c\left(u_1,\tau_{\gamma_s}(u_1)\right)\\
	&\quad\times f_2\left(F_2^{-1}\left(\tau_{\gamma_s}(u_1)\right)\right)du_1d\gamma_s,\label{ac-gen}
\end{align}
where $c(u_1,\tau_{\gamma_s}(u_1))=\frac{\partial^2C(u_1,\tau_{\gamma_s}(u_1))}{\partial u_1\partial \tau_{\gamma_s}(u_1)}$, and 
\begin{align}\nonumber
	&\tau_{\gamma_s}(u_1)=\\ 
	&I_{1-\frac{1}{\frac{m_2}{m_{2,s}} \left(\gamma_s-\frac{m_{1,s}}{m_1}\left(\frac{1}{I_{(1,-u_1)}^{-1}(m_{1,s},m_1)}-1\right)\right)+m_{2,s}}}(m_2,m_{2,s}).
\end{align}
\end{theorem}
\begin{proof} By assuming $\gamma_s=\gamma_1+\gamma_2$ and inserting   $\tau_{\gamma_s}(u_1)=F_2\left(\gamma_s-F_1^{-1}(u_1)\right)$, and $F_2^{-1}(\tau_{\gamma_s}(u_1))=\gamma_s-F_1^{-1}(u_1)$ into \eqref{pdf}, the PDF $f_s(\gamma_s)$ is obtained. Then, by substituting $f_s(\gamma_s)$ into the AC definition for the fading clean MAC, \eqref{asc-def}, the proof is completed. 
\end{proof}
\begin{corollary}
The AC over correlated Fisher-Snedecor $\mathcal{F}$ fading clean MAC, using the Clayton copula is given by
\begin{align}\nonumber
	&\bar{\mathcal{C}}^{CM}
	=\frac{(1+\theta)}{2}\int_{0}^{\infty}\int_{0}^{1}\log_2\left(1+\gamma_s\right)\left(u_1\tau_{\gamma_s}(u_1)\right)^{-1-\theta}\\
	&\times\left({u_1}^{-\theta}+{\tau_{\gamma_s}^{-\theta}(u_1)}-1\right)^{-2-\frac{1}{\theta}}
 f_2\left(F_2^{-1}\left(\tau_{\gamma_s}(u_1)\right)\right)du_1d\gamma_s.
\end{align}
\end{corollary}
\begin{proof} By computing the copula density function $c_{cl}$ for the Clayton copula as follows:
	\begin{align}
	c_{\mathrm{cl}}(u_1,u_2)=\left(1+\theta\right)\left(u_1u_2\right)^{-1-\theta}\left({u_1}^{-\theta}+{u_2}^{-\theta}-1\right)^{-2-\frac{1}{\theta}},\label{cop-dens}
	\end{align}
and then, considering $\tau_{\gamma_s}(u_1)$ instead of $u_2$ in \eqref{cop-dens} and substituting \eqref{cop-dens} into \eqref{ac-gen}, the proof is completed.
\end{proof}
\subsubsection{Doubly dirty MAC} Similarly, with the same assumption, the AC for the considered correlated fading doubly dirty MAC is obtained as the following theorem. 
\begin{theorem}
The AC over correlated Fisher-Snedecor $\mathcal{F}$ fading doubly dirty MAC and any arbitrary copula density function $c$ is determined as:
\begin{align}\nonumber
\bar{\mathcal{C}}^{DM}=&\int_{0}^{\infty}\int_{0}^{\infty}\frac{1}{2}\log_2\left(1+\min\left\{\gamma_1,\gamma_2\right\}\right)\\
&\times f_1(\gamma_1)f_2(\gamma_2)c\left(F_1(\gamma_1),F_2(\gamma_2)\right)d\gamma_1d\gamma_2,\label{ac-dd-cor}
\end{align}
\end{theorem}
\begin{proof} By applying the chain rule to \eqref{sklar}, the joint PDF $f(\gamma_1,\gamma_2)$ is determined as
	\begin{align}
	f(\gamma_1,\gamma_2)&=\frac{\partial^2 C\left(F_1(\gamma_1),F_(\gamma_2)\right)}{\partial \gamma_1\partial \gamma_2},\\
	&=\frac{\partial^2 C\left(F_1(\gamma_1),F_(\gamma_2)\right)}{\partial F_1(\gamma_1)\partial F_2(\gamma_2)}\frac{\partial F_1(\gamma_1)}{\partial \gamma_1}\frac{\partial F_2(\gamma_2)}{\partial\gamma_2},\\
	&=c\left(F_1(\gamma_1),F_2(\gamma_2)\right)f_1(\gamma_1)f_2(\gamma_2).\label{copula-density}
	\end{align}
Now, by applying \eqref{copula-density} to the AC definition, the proof is completed. 
	\end{proof}
\begin{corollary}
The AC over correlated Fisher-Snedecor $\mathcal{F}$ fading doubly dirty MAC, using the Clayton copula is given by
\begin{align}\nonumber
&\hspace{-0.2cm}\bar{\mathcal{C}}^{DM}=\frac{\left(1+\theta\right)}{2}\int_{0}^{1}\int_{0}^{1}\log_2\left(1+\min\right\{F_1^{-1}(u_1),F_2^{-1}(u_2)\left\}\right)\\
&\quad\quad\,\times \left(u_1u_2\right)^{-1-\theta}\left({u_1}^{-\theta}+{u_2}^{-\theta}-1\right)^{-2-\frac{1}{\theta}}du_1du_2,\label{ac-dd-cl}
\end{align}
where $F_i^{-1}(u_i)=\frac{m_{i,s}}{m_i} \left(\frac{1}{I_{(1,-u_i)}^{-1}(m_{i,s},m_i)}-1\right)$ for $i\in\{1,2\}$.
\end{corollary}
\begin{proof}
By substituting the Clayton copula density function, \eqref{cop-dens}, into $\eqref{ac-dd-cor}$, the AC is derived as \eqref{ac-dd-cl}.
\end{proof}
\section{Numerical Results}\label{num-results}\vspace{-0.1cm}
\begin{figure*}
	\centering
	\subfigure[Clean MAC]{%
		\includegraphics[width=0.26\textwidth]{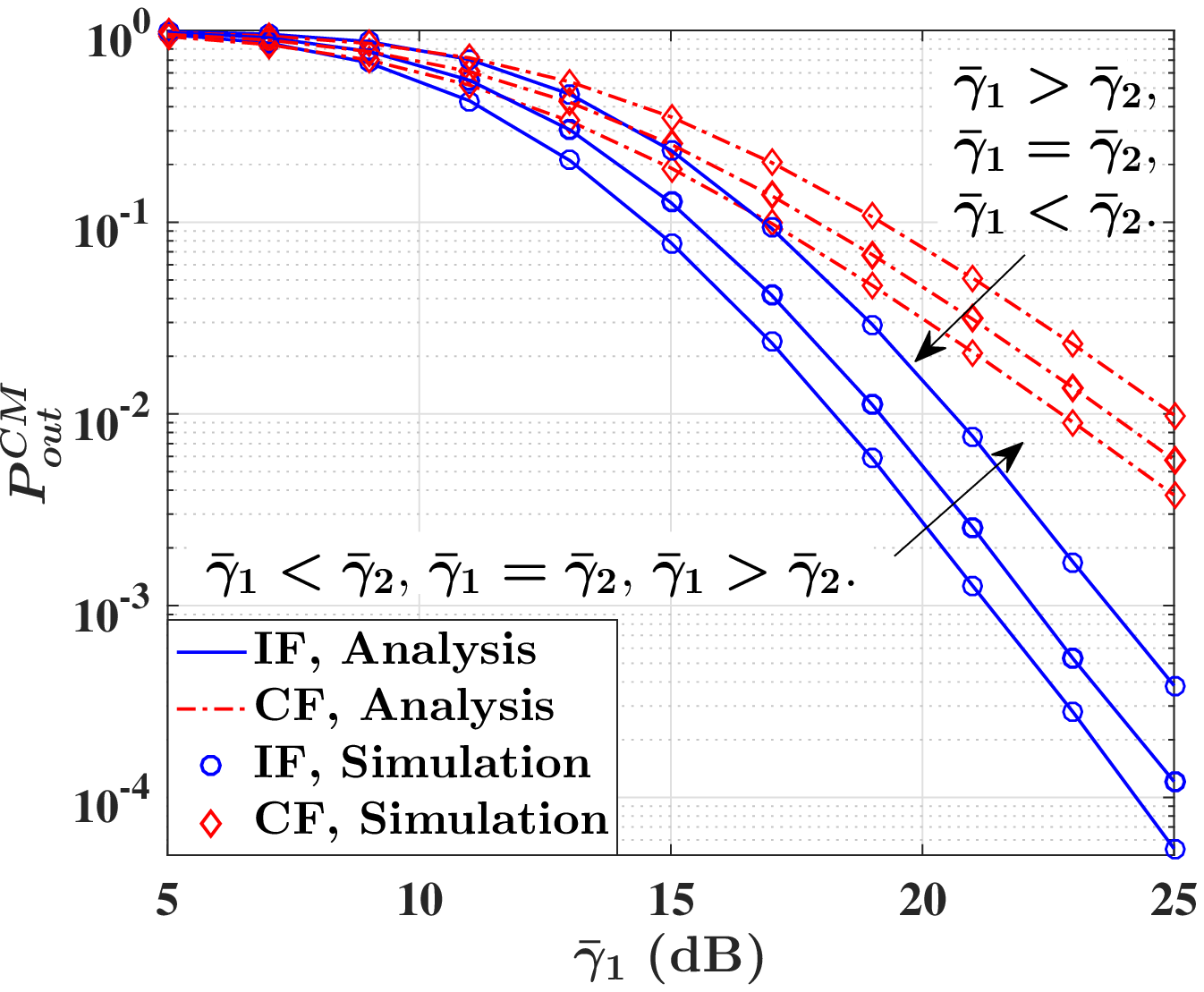}%
		\label{fig-mac}%
	}\hspace{-0.37cm}
	\subfigure[Doubly dirty MAC]{%
		\includegraphics[width=0.26\textwidth]{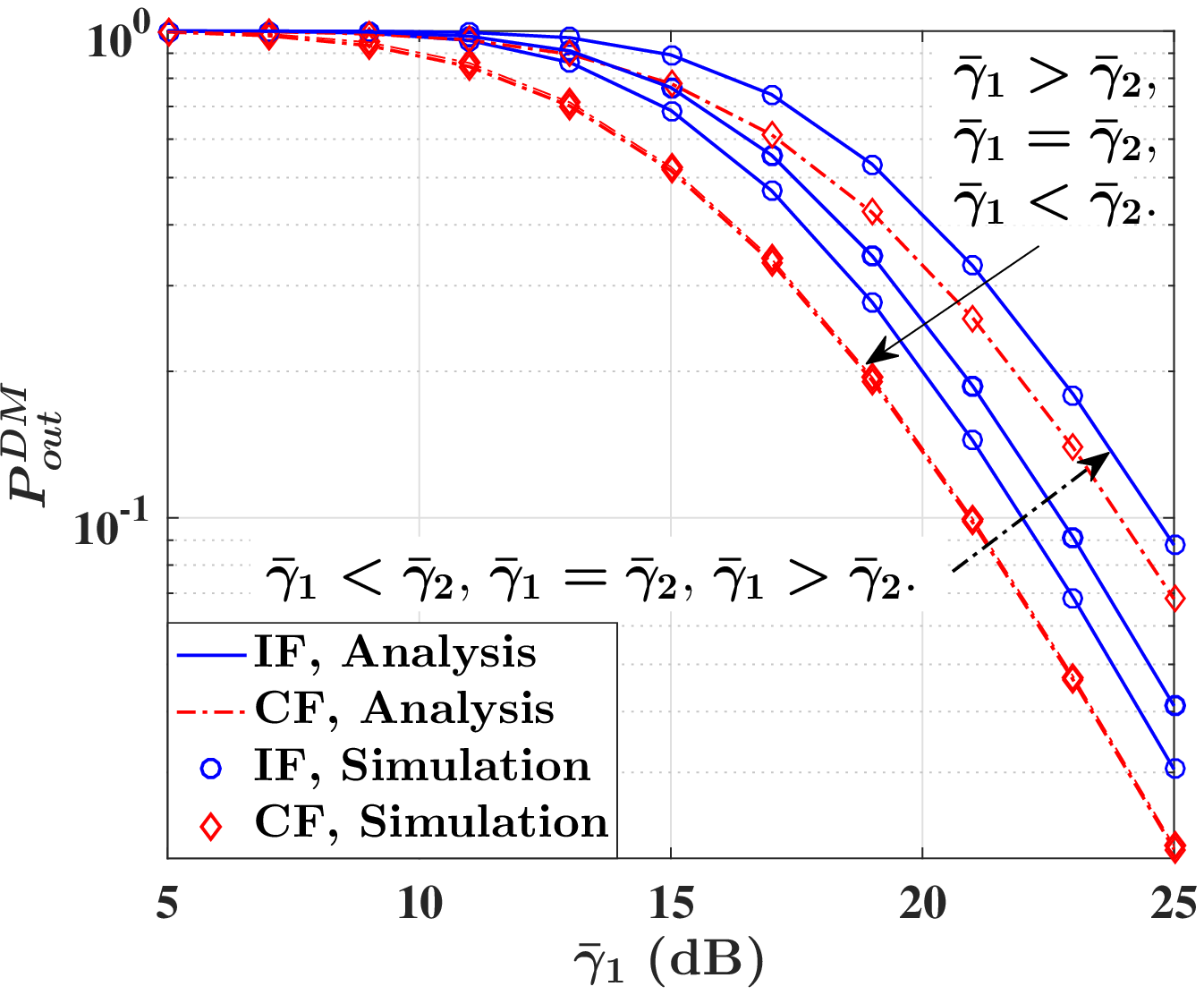}%
		\label{fig-side}%
	}%
	\subfigure[Clean MAC]{%
		\includegraphics[width=0.26\textwidth]{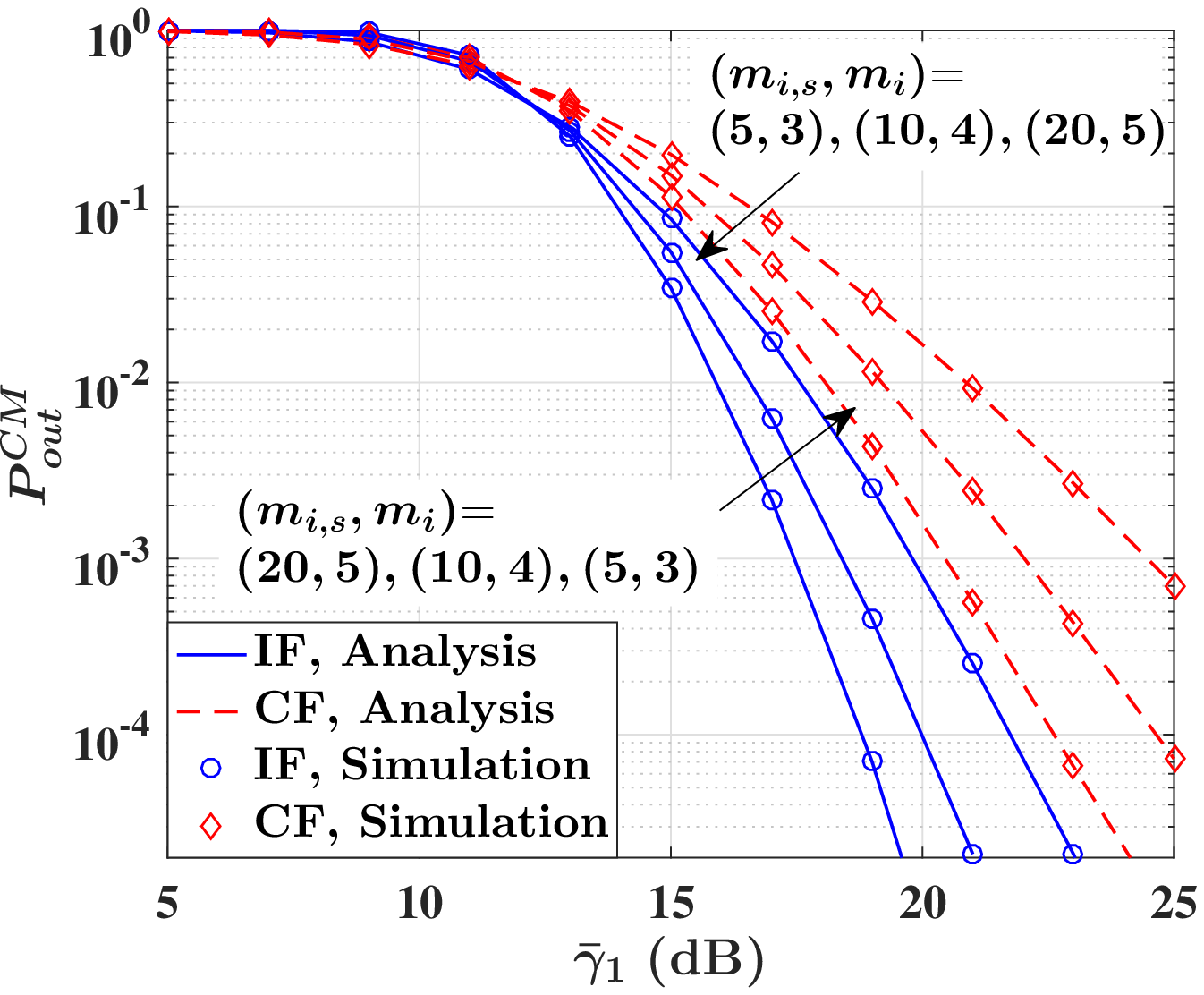}%
		\label{fig-mac-m}%
	}\hspace{-0.37cm}
	\subfigure[Doubly dirty MAC]{%
		\includegraphics[width=0.26\textwidth]{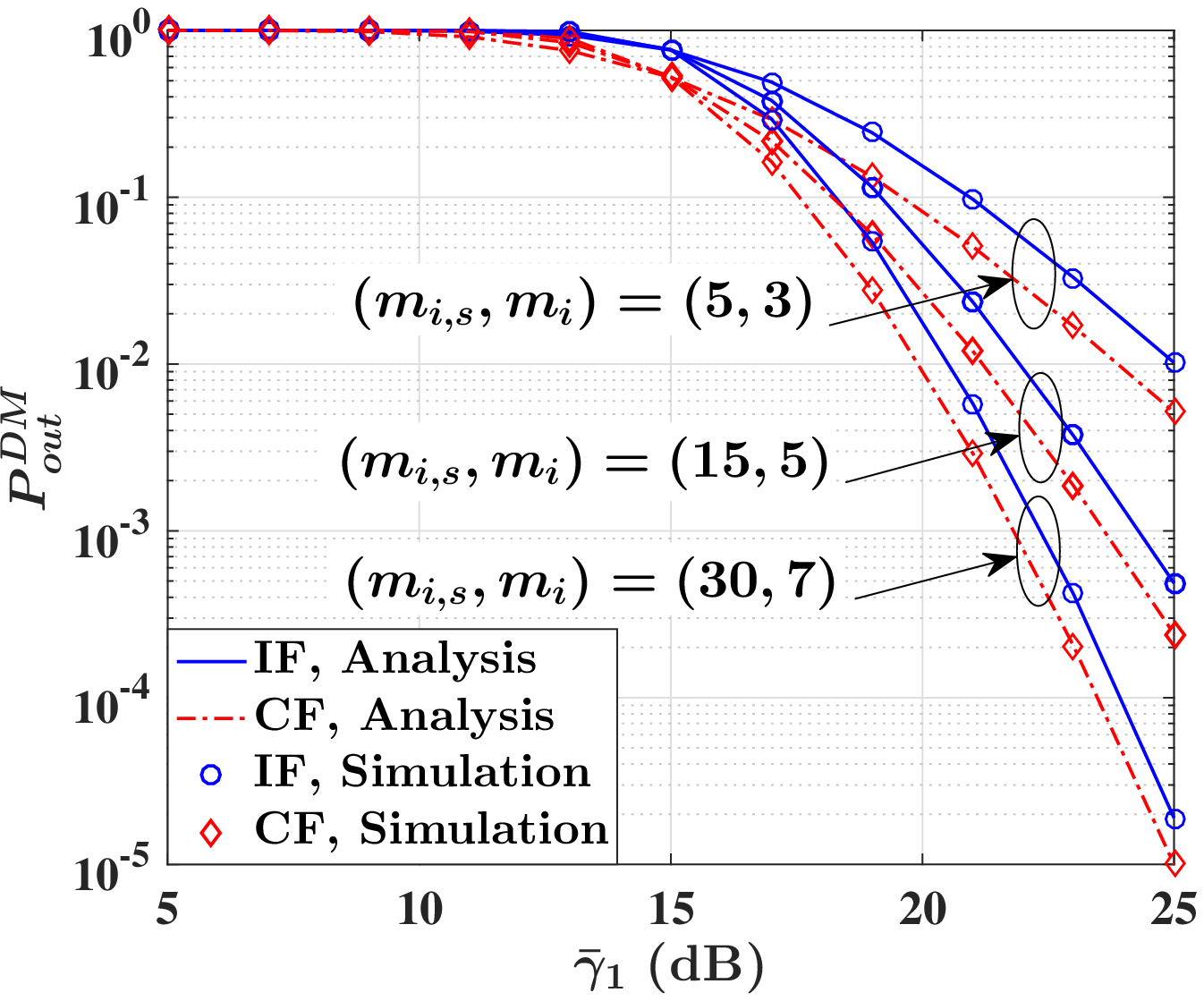}%
		\label{fig-side-m}%
	}\hspace{0cm}
	\caption{The OP versus average SNR $\bar{\gamma}_1$ over Fisher-Snedecor $\mathcal{F}$ fading MAC when: $R_t=2.5$, $m_{i,s}=3, m_i=2$, and $\theta=40$ for (a) and (b); $R_t=2.5$ and $\theta=40$ for (c) and (d).}\vspace{-0.6cm}
	\label{fig:ab}
\end{figure*}
\begin{figure*}
	\centering
	\subfigure[Clean MAC]{%
		\includegraphics[width=0.26\textwidth]{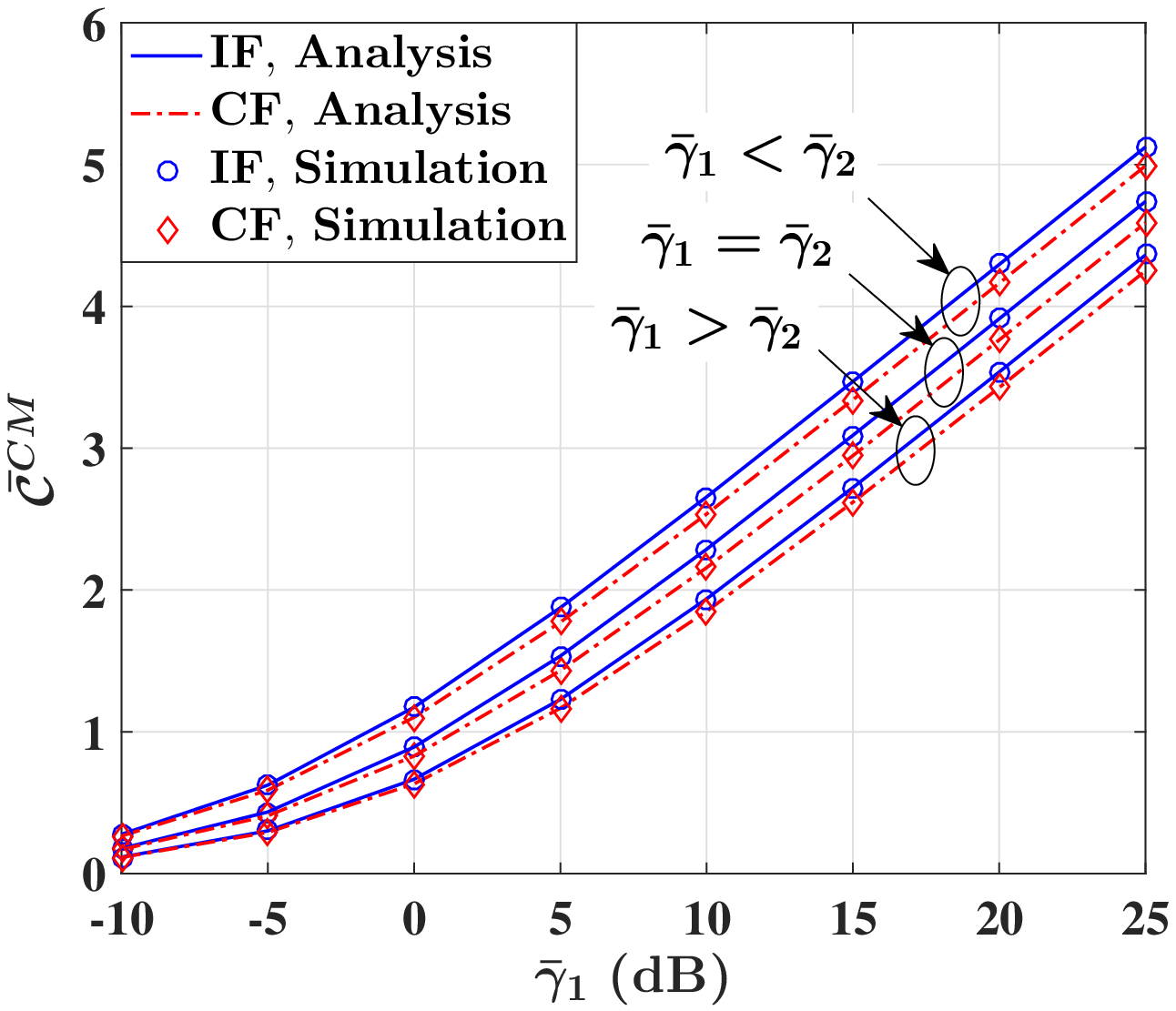}%
		\label{fig-ac-mac}%
	}\hspace{-0.37cm}
	\subfigure[Doubly dirty MAC]{%
		\includegraphics[width=0.26\textwidth]{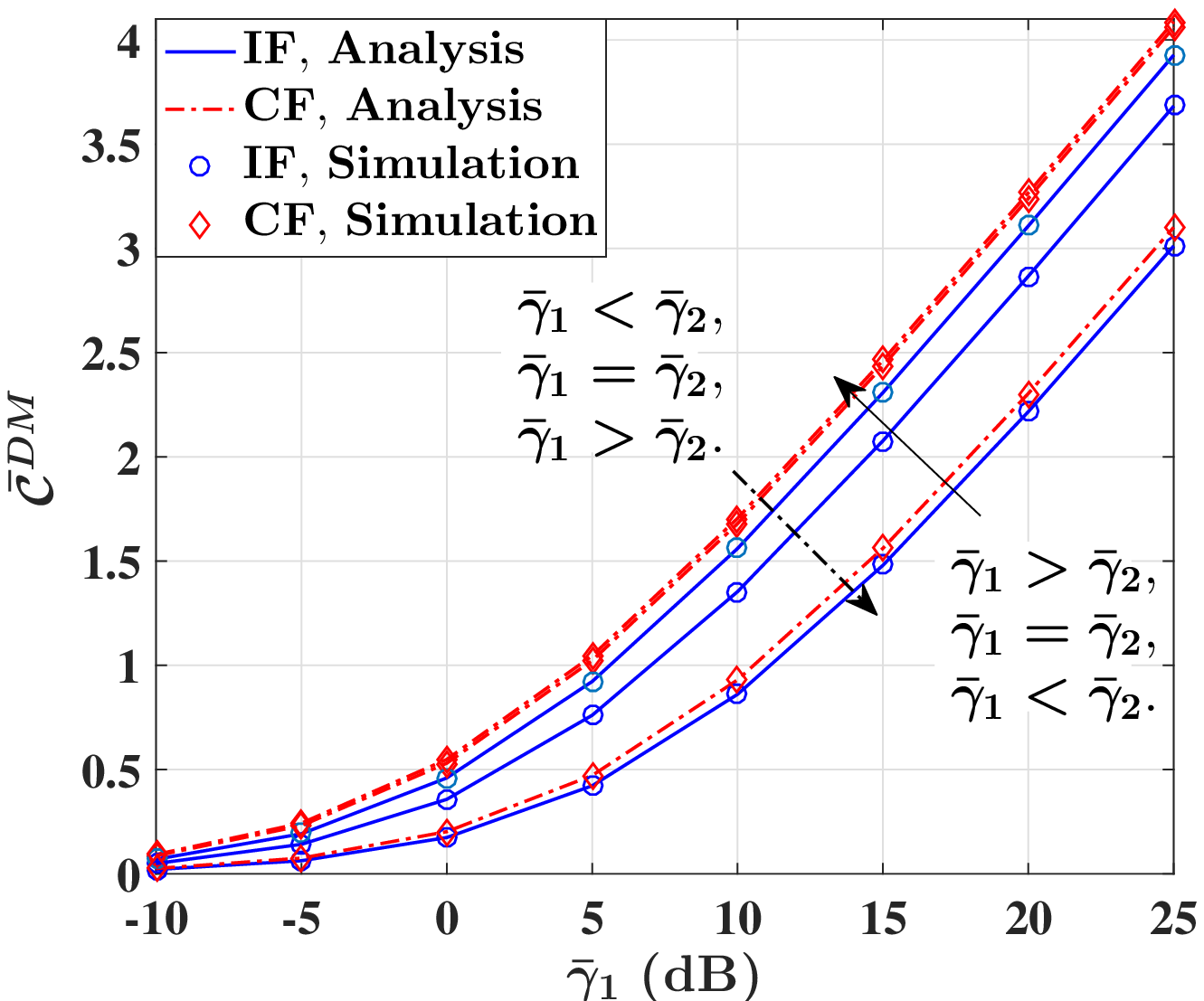}%
		\label{fig-ac-side}%
	}%
	\subfigure[Clean MAC]{%
		\includegraphics[width=0.26\textwidth]{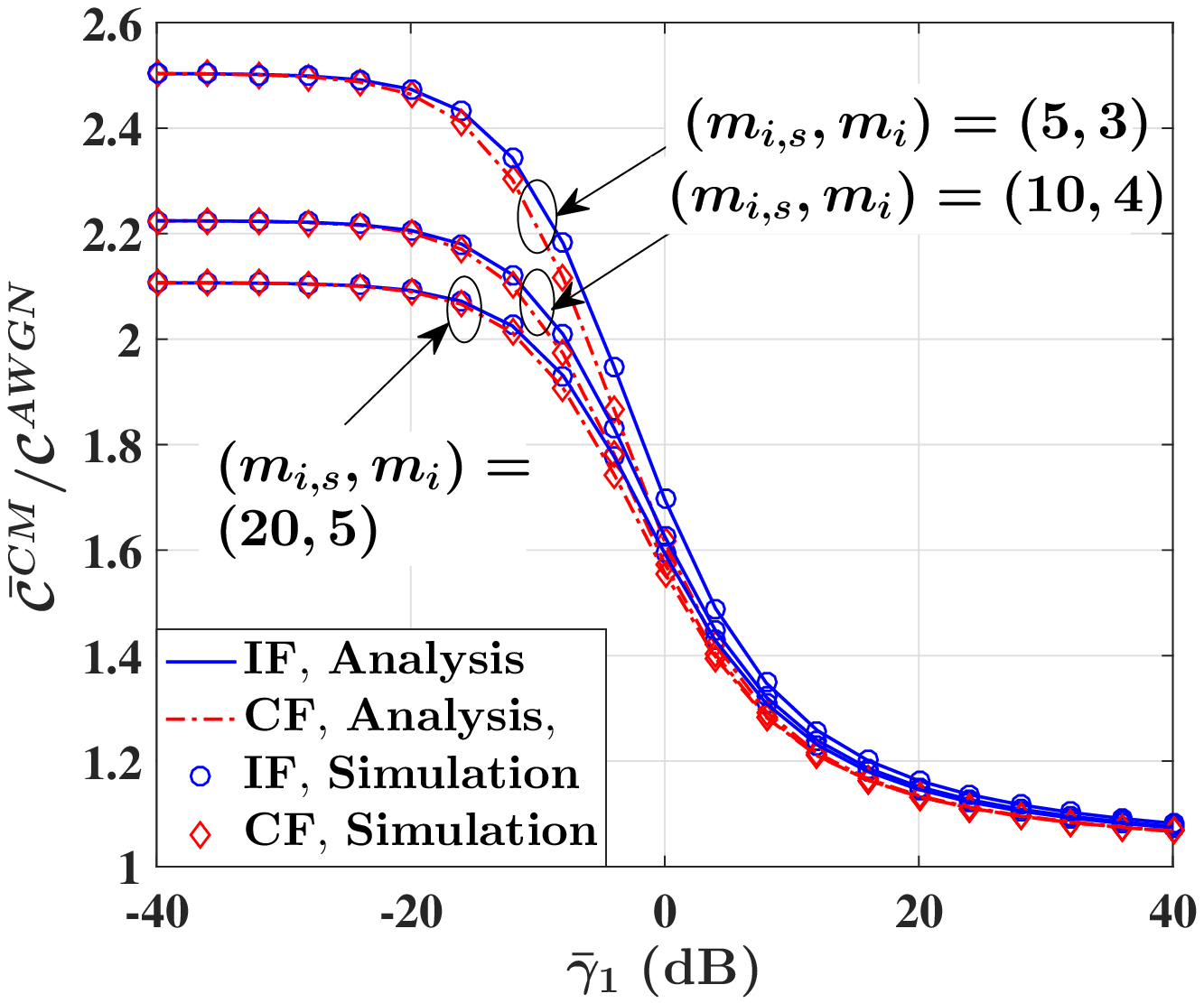}%
		\label{fig-awgn-mac}%
	}\hspace{-0.37cm}
	\subfigure[Doubly dirty MAC]{%
		\includegraphics[width=0.26\textwidth]{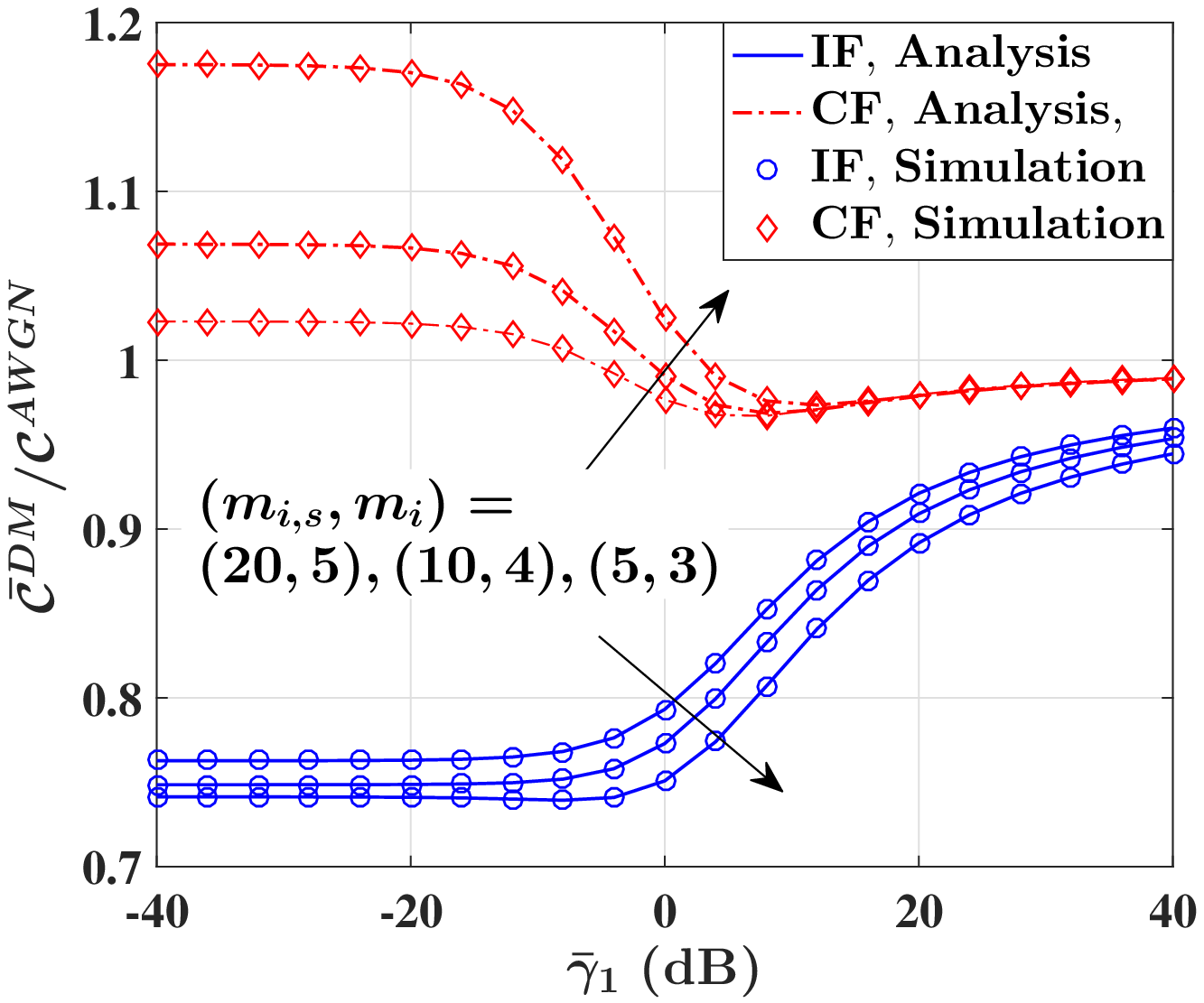}%
		\label{fig-awgn-side}%
	}\hspace{0cm}
	\caption{Capacity metrics versus $\bar{\gamma}_1$ over Fisher-Snedecor $\mathcal{F}$ fading MAC in different scenarios: (a), (b) AC versus $\bar{\gamma}_1$ when  $R_t=2.5$, $m_{i,s}=3, m_i=2$, and $\theta=40$; (c), (d) AC normalized to that of the AWGN case versus $\bar{\gamma}_1$ when $R_t=2.5$ and $\theta=40$.}\vspace{-0.3cm}
	\label{fig-c}
\end{figure*}
In this section, we evaluate the theoretical expressions previously
derived, which are double-checked in all instances with Monte Carlo (MC) simulations. We also set fading parameters $(m_{i,s},m_i)$ based on the results obtained in \cite{yoo2017fisher}. Besides, it should be noted that although the implementation of the extended generalized bivariate Meijer's G-function is not available in mathematical packages, like Mathematica, Maple, or MATLAB, it is computationally tractable and programmable as explained in \cite{peppas2012new}. Figs. \ref{fig-mac} and \ref{fig-side} represent the performance of OP over Fisher-Snedecor $\mathcal{F}$ clean MAC and doubly dirty MAC based on the variation of $\bar{\gamma}_1$ in the presence/absence of fading correlation for the fixed values of fading parameters, dependence parameter, and three different scenarios (i.e., $\bar{\gamma}_1<\bar{\gamma}_2$, $\bar{\gamma}_1=\bar{\gamma}_2$, and $\bar{\gamma}_1>\bar{\gamma}_2$), respectively. In all instances, it can be seen that the OP constantly decreases by increasing $\bar{\gamma}_1$, which is reasonable because the channel condition is improved. In Fig. \ref{fig-mac}, we can see that the independent fading (IF) case provides better performance as compared with the correlated fading (CF) case in terms of the OP in the clean MAC model. In contrast, Fig. \ref{fig-side} shows that the CF case has achieved better performance compared to the IF case in doubly dirty MAC.
The behavior of OP in terms of $\bar{\gamma}_1$ over Fisher-Snedecor $\mathcal{F}$ clean MAC and doubly dirty MAC for selected values of fading parameters and $\theta=40$ is illustrated in Figs. \ref{fig-mac-m} and \ref{fig-side-m}, respectively. In both models, we can see that as the shadowing and the fading are less severe, i.e., as $(m_{i,s},m_i)$ for $i\in\{1,2\}$ increase, the performance of OP is improved. In other words, the communication performance improves (degrades) in environments that exhibit light (heavy) shadowed fading characteristics. Figs. \ref{fig-ac-mac} and \ref{fig-ac-side} show the behavior of the AC under correlated/independent Fisher-Snedecor $\mathcal{F}$ fading clean MAC and doubly dirty MAC with $(m_{i,s},m_i)=(3,2)$. In Fig. \ref{fig-ac-mac}, we can see that the positive dependence structure is detrimental to the AC performance over clean MAC. In contrast, as shown in Fig. \ref{fig-ac-side}, the CF case provides higher values of the AC as opposed to the IC case in the doubly dirty MAC model. Given that capacity in the low/high SNR regimes highly depends on the fading severity and the dependency of fading coefficients, we now normalize the AC to that of the AWGN case for both clean and doubly dirty MAC models. From Fig. \ref{fig-awgn-mac}, it becomes evident that CF under positive dependence structure has destructive effects on the capacity performance compared to the IF case in the clean MAC model. Contrastingly, as shown in Fig. \ref{fig-awgn-side}, the CF provides a larger capacity in the doubly dirty MAC model, meaning that the IF is harmful to the capacity performance. We also see that this destructive effect becomes more noticeable under a strong fading condition (i.e., $(m_{i,s},m_i)=(5,3)$) than when a milder one (i.e., $(m_{i,s},m_i)=(20,5)$) is considered. It should be noted that the bottleneck effect imposed by the transmitters with a minimum SNR in the capacity region \eqref{cap-si}, is relaxed in the presence of fading correlation. Besides, it is important to highlight that this is in stark contrast with the observations made in clean MAC in the absence of interference, for which the opposite conclusion was obtained. Therefore, we see that considering the non-causally known SI at transmitters in the clean MAC (i.e., doubly dirty MAC) can improve the performance of OP and AC under the positive dependence structure.

In order to gain more insights into the effect of dependence structure in Fisher-Snedecor $\mathcal{F}$ fading, Table \ref{tab} provides the measure of dependency between fading coefficients in terms of the correlation coefficient $\rho\overset{\Delta}{=}\mathrm{cov}[\gamma_1\gamma_2]/\sqrt{\mathrm{var}[\gamma_1]\mathrm{var}[\gamma_2]}$, by simulating the dependence parameter of Clayton copula $\theta$. To this end, we consider two scenarios: (\textit{a}) the fading parameters are equal for both channels; (\textit{b}) the fading parameters are different for each channel (i.e., one of the channels experiences more shadowing and fading severity). For the first scenario, we present four cases: (\textit{i}) both $m_i$ and $m_{i,s}$ are fixed; (\textit{ii}) $m_i$ are fixed and $m_{i,s}$ are varied; (\textit{iii}) $m_i$ are varied and $m_{i,s}$ are fixed; (\textit{iv}) both $m_i$ and $m_{i,s}$ are varied. For the second scenario, $m_i$ and $m_{i,s}$ are different so that the first channel is under heavier shadowing and fading environment. For all cases in both scenarios, it can be seen that the dependence structure of the channel coefficients in Fisher-Snedecor $\mathcal{F}$ fading highly depends on the fading and copula parameters, where the correlation coefficient $\rho$ is increased as $\theta$ and $(m_i,m_{i,s})$ grows. Noteworthy, 
when $\theta$ goes to larger values (e.g., $\theta=40$) 
the perfect positive correlation (i.e., $\rho\rightarrow 1$) is reached. Furthermore, the dependence between SNRs $\gamma_1$ and $\gamma_2$ in Table \ref{tab} can be intuitively obtained from a scatter plot, i.e., representing the realizations of $\gamma_1$ versus those of $\gamma_2$. For informative purposes, we represent the scatter plots corresponding to independent and positive dependence structures between $\gamma_1$ and $\gamma_2$ using Clayton copula in Fig. \ref{fig-scatter}. It can be seen the data scattering decreases as $\theta$ grows, meaning that a stronger correlation is reached for higher values of $\theta$.  We can also see that the Clayton copula can efficiently justify the heavy concentration in the left tail, which indicates that the Clayton copula is an appropriate choice for performance analysis of the proposed model due to the fact that deep fade happens in tails. \vspace{0cm}
\begin{figure*}
	\centering
	\subfigure[$\theta\rightarrow 0$]{%
		\includegraphics[width=0.26\textwidth]{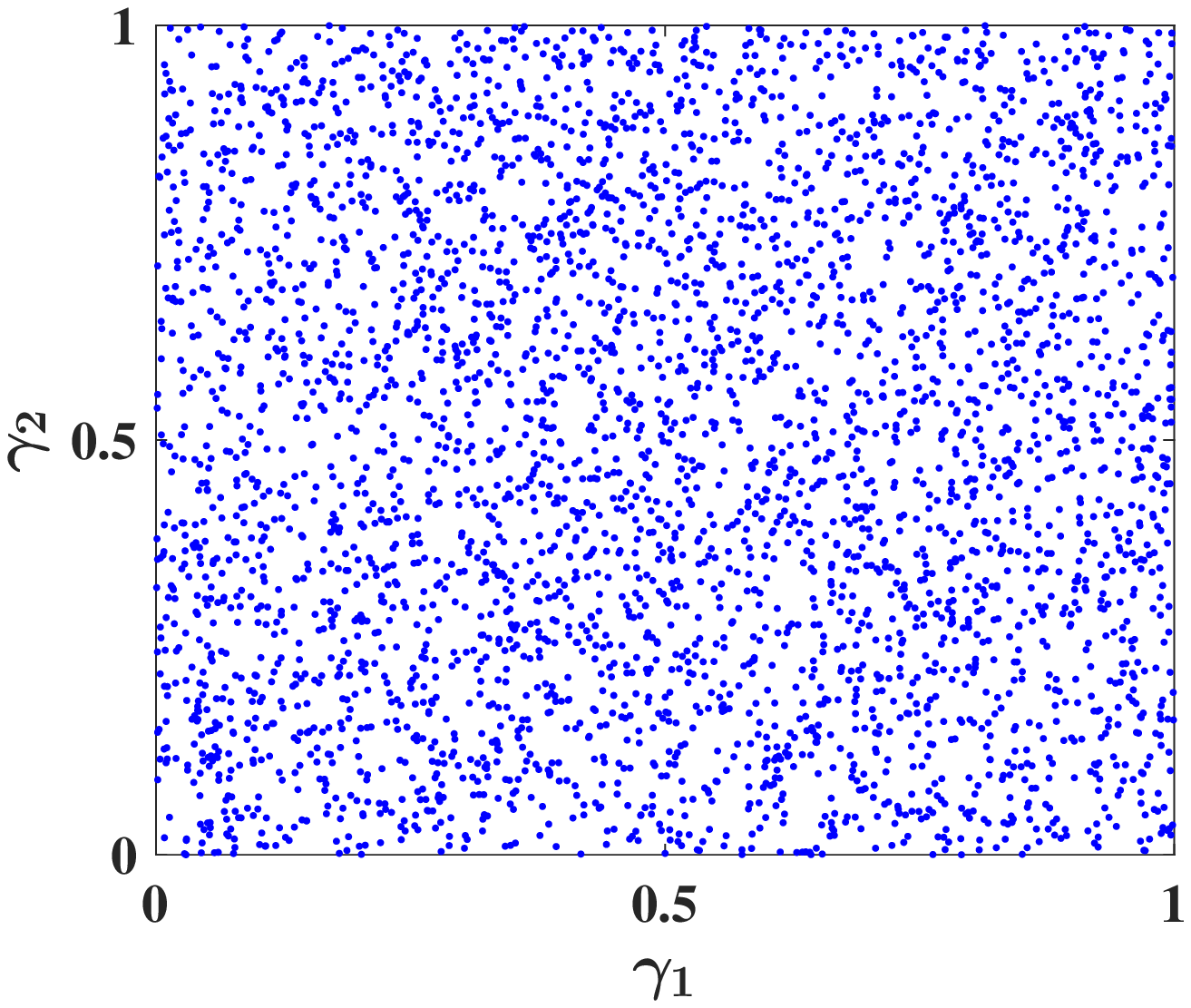}%
		\label{}%
	}\hspace{-0.37cm}
	\subfigure[$\theta=10$]{%
		\includegraphics[width=0.26\textwidth]{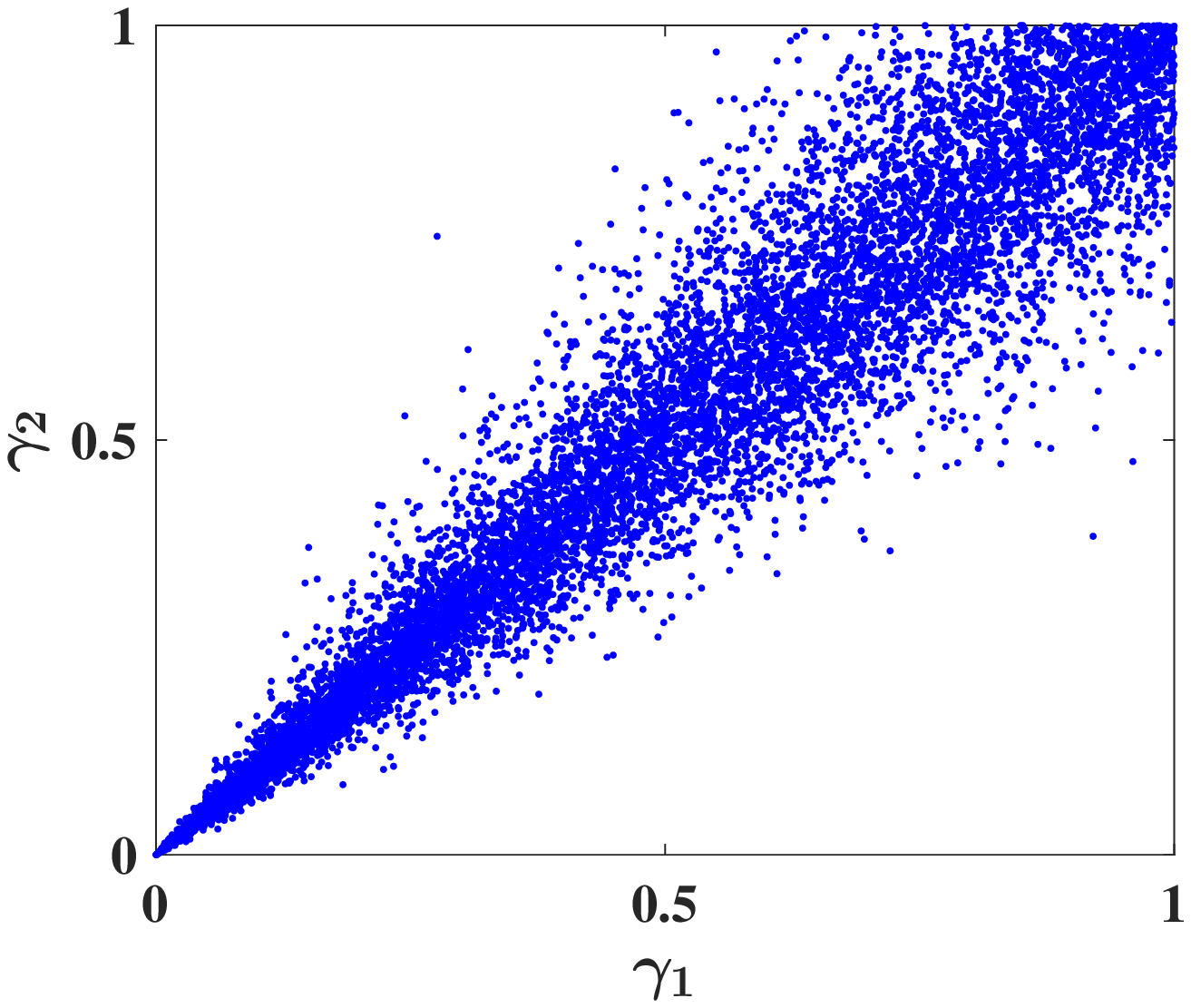}%
		\label{}%
	}\hspace{-0.37cm}
	\subfigure[$\theta=25$]{%
		\includegraphics[width=0.26\textwidth]{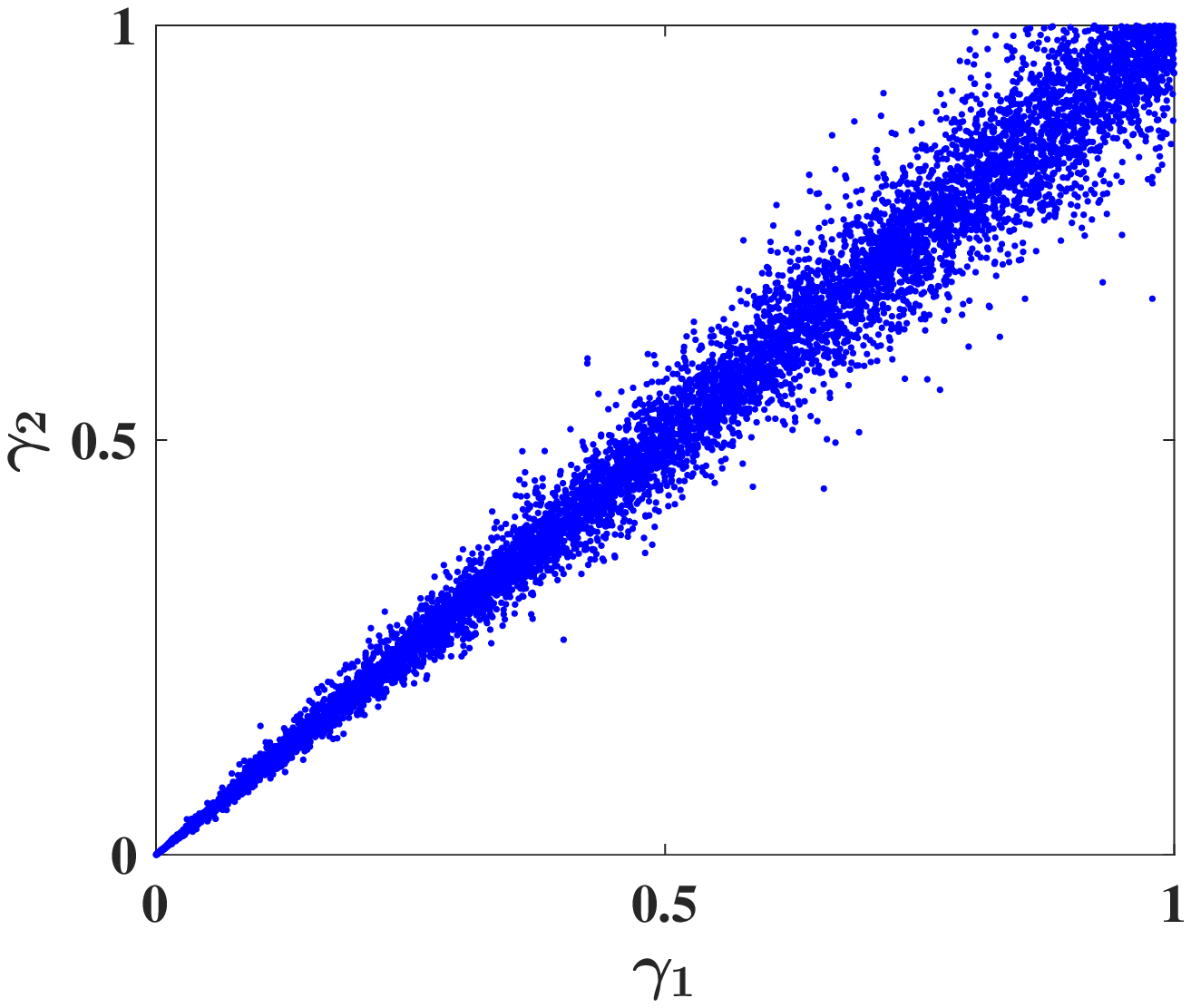}%
		\label{}%
	}\hspace{-0.37cm}
	\subfigure[$\theta=40$]{%
		\includegraphics[width=0.26\textwidth]{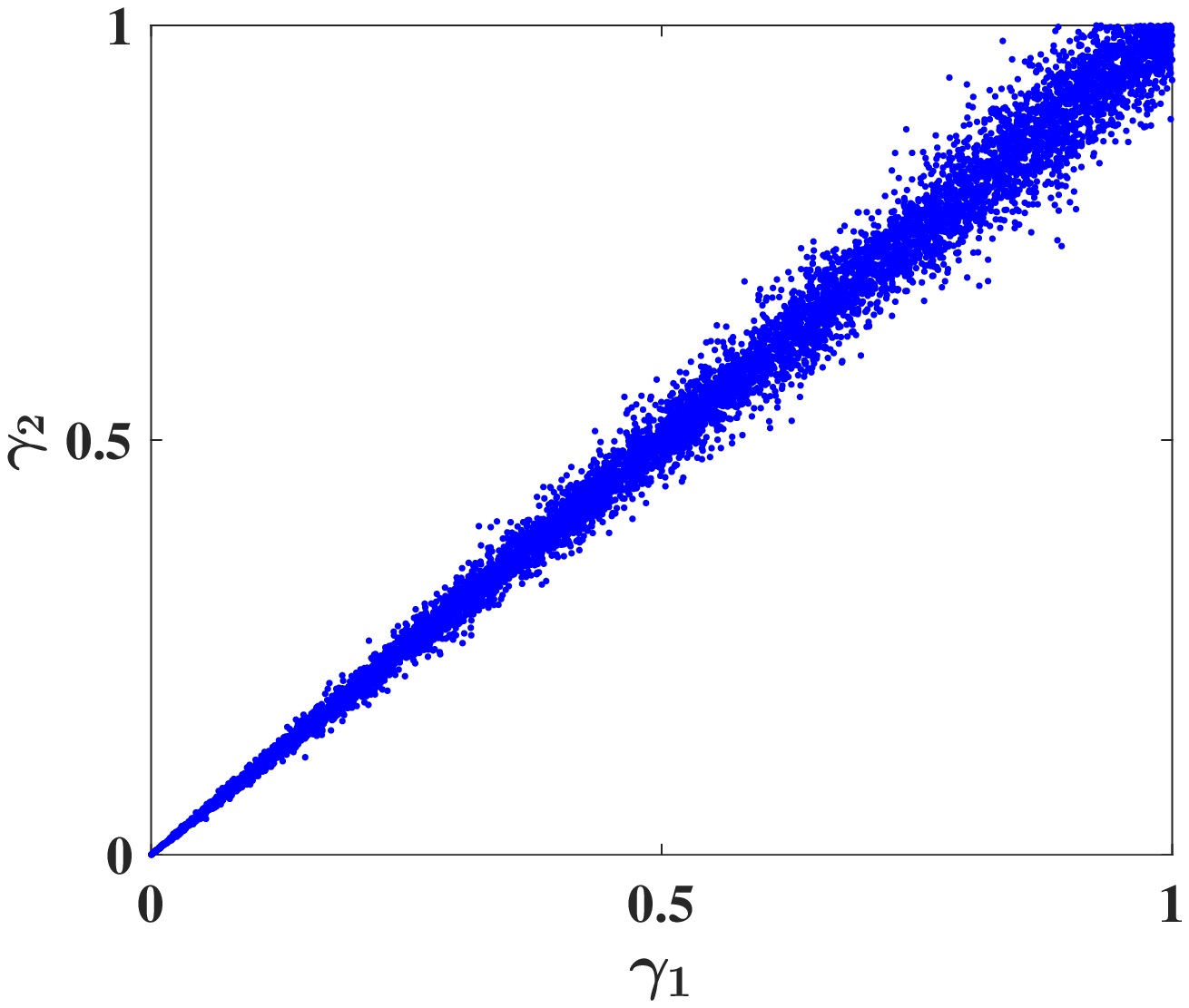}%
		\label{}%
	}%
	\caption{Scatter plots with independent and positive dependence between SNRs $\gamma_1$ and $\gamma_2$
			under Clayton copula when $(m_{i,s}; m_i) = (5,3)$: (a) $\theta\rightarrow 0$ (i.e., independent); (b) $\theta=10$; (c)
			$\theta=25$; (d) $\theta=40$.}\label{fig-scatter}\vspace{-0.5cm}
\end{figure*}
\begin{table}\scriptsize
\begin{center}
			\caption{The values of correlation coefficients $\rho$ in terms of fading parameters $(m_{i,s},m_i)$ for selected values of the dependence parameter $\theta$ under two different scenarios (\textit{a}) and (\textit{b})}.\vspace{-0.5cm}
	\begin{tabular}{ |c|c|c|c|c|c|c|c| }
		\hline
		\hspace{-1ex}Scenario (\textit{a})\hspace{-1ex} & $m_i$ & $m_{i,s}$ & $\theta$ & $\rho$ \\
		\hline
		&2 & 2 & 10, 25, 40  & 0.1300, 0.2293, 0.2812  \\ 
		(\textit{i}) &5 & 5 & 10, 25, 40 & 0.6982, 0.8149, 0.8574 \\ 
		&7 & 7 & 10, 25, 40 & 0.7696, 0.8732, 0.9084 \\ 
		\hline
		&2 & 3 & 10, 25, 40  & 0.4791, 0.6107, 0.6581  \\ 
		(\textit{ii}) &2 & 5 & 10, 25, 40 & 0.6796, 0.7995, 0.8493 \\ 
		&2 & 20 & 10, 25, 40 & 0.8108, 0.9092, 0.9383 \\ 
		\hline
		&3 & 3 & 10, 25, 40 & 0.4380, 0.6184, 0.6637 \\ 
		(\textit{iii}) &5 & 3 & 10, 25, 40 & 0.4396, 0.6285, 0.6665  \\
		&7 & 3 & 10, 25, 40 & 0.4884, 0.6358, 0.6822  \\
	    \hline
	    &3 & 5 & 10, 25, 40 & 0.6916, 0.8151, 0.8568  \\
	    (\textit{iv}) &5 & 15 & 10, 25, 40 & 0.8316, 0.9205, 0.9467  \\
	    &7 & 30 & 10, 25, 40 & 0.8664, 0.9410, 0.9624  \\
	    \hline
	    \hspace{-1ex}Scenario (\textit{b})\hspace{-1ex} &\hspace{-1.5ex}($m_{1,s},m_{1}$)\hspace{-1.5ex} & \hspace{-1.5ex}($m_{2,s},m_2$)\hspace{-1.5ex} & 10, 25, 40 & 0.8664, 0.9410, 0.9624  \\
	    \hline
	    &(3,2) & (5,3) & 10, 25, 40 & 0.5780, 0.7045, 0.7471  \\
	    (\textit{b}) &(5,3) & (15,5) & 10, 25, 40 & 0.7528, 0.8513, 0.8847  \\
	    &(15,5) & (30,7) & 10, 25, 40 & 0.8473, 0.9277, 0.9526  \\
	    \hline
	\end{tabular}
\label{tab}
\end{center}\vspace{0cm}
\end{table}\vspace{0cm}
\section{Conclusion}\label{conclusion}\vspace{0cm}
In this paper, we analyzed the performance of wireless multiple access communication systems under correlated/independent Fisher-Snedecor $\mathcal{F}$ fading conditions. To this end, we derived the exact analytical expressions for the OP and the AC in the correlated/independent fading clean and doubly dirty MAC scenarios, exploiting copula theory. Besides, to get a clearer insight into the measure of fading dependence, we studied the correlated fading case in 
the positive dependence structure applying the Clayton copula. The results showed that considering the fading correlation has a destructive effect on the performance of clean MAC in terms of the OP and the AC, while in doubly dirty MAC which transmitters include non-causally known SI, the performance of OP and AC improves under correlated fading conditions for positive dependence structures. Moreover, regarding the capability of copula theory in describing the dependence structures of fading channels, analyzing other multi-user communication systems such as interference channel (IC) will be significant as future works. \vspace{0cm} 
	\appendices
	\section{Proof of Theorem \ref{p-out-mac}}\label{app-out-mac}
	By exploiting the Parseval relation for Mellin transform \cite[Eq. (8.3.23)]{debnath2014integral}, \eqref{app-out} can be rewritten as follows:
	\begin{align}
		P_{out}^{CM}&=\int_{0}^{\infty}F_1(\gamma_0)f_2({\gamma_2})d\gamma_2,\\
		&=\frac{1}{2\pi j}\int_{\mathcal{L}_1}\mathcal{M}\left[F_1(\gamma_0),1-s\right]\mathcal{M}\left[f_2(\gamma_2),s\right]ds,\label{ap3}
	\end{align}
	where $\mathcal{L}_1$ is the integration path from $\nu-j\infty$ to $\nu+j\infty$ for a constant value of $\nu$ \cite{yoo2017fisher}. Then, by exploiting the definition of Meijer's G-function, we have:
	\begin{align}
		&\mathcal{M}\left[F_1(\gamma_0),1-s\right]=\int_{0}^{\infty}{\gamma_2}^{-s}F_1({\gamma_0})d\gamma_2,\\\nonumber
		&=\frac{\mathcal{B}_1}{2\pi j}\int_{\mathcal{L}_2}\frac{\Gamma(-\zeta)\Gamma(m_1+\zeta)\Gamma(m_{1,s}-\zeta)}{\Gamma(1-\zeta){\lambda_1}^\zeta}d\zeta\\
		&\;\;\;\;\times\int_{0}^{\infty}{\gamma_2}^{-s}{\gamma_0}^{-\zeta}d\gamma_2,\label{ap1}
	\end{align}\vspace{0cm}
	where, the inner integral can be obtained as:
	\begin{align}
		&\int_{0}^{\infty}{\gamma_2}^{-s}{\gamma_0}^{-\zeta}d\gamma_2\overset{(a)}{=}{\gamma_t}^{-\zeta}\int_{0}^{\infty}\frac{{\gamma_2}^{-s}}{\left(1-\frac{\gamma_2}{\gamma_t}\right)^\zeta}d\gamma_2,\\
		&\overset{(b)}{=}{\gamma_t}^{-\zeta}\mathit{B}\left(1-s,\zeta+s-1\right)\left(\frac{-1}{\gamma_t}\right)^{s-1},\\
		&\overset{(c)}{=}{\gamma_t}^{-\zeta}\frac{\Gamma(1-s)\Gamma(\zeta+s-1)}{\Gamma(\zeta)}\left(\frac{-1}{\gamma_t}\right)^{s-1},\label{ap2}
	\end{align}
	where $(a)$ is obtained representing $\gamma_0=\gamma_t-\gamma_2$, $(b)$ is derived form \cite[eq. (3.194.3)]{gradshteyn2014table}, and $(c)$ is obtained by utilizing the property of beta function where $B(a_1,b_1)=\frac{\Gamma(a_1)\Gamma(b_1)}{\Gamma(a_1+b_1)}$. Thus, by inserting \eqref{ap2} into \eqref{ap1}, we have:
	\begin{align}\nonumber
		&\mathcal{M}\left[F_1(\gamma_0),1-s\right]=\frac{\mathcal{B}_1\Gamma(1-s)}{2\pi j}\left(\frac{-1}{\gamma_t}\right)^{s-1}\\
		&\times\int_{\mathcal{L}_2}\frac{\Gamma(-\zeta)\Gamma(m_1+\zeta)\Gamma(m_{1,s}-\zeta)\Gamma(\zeta+s-1)}{\Gamma(\zeta)\Gamma(1-\zeta){\left(\gamma_t\lambda_1\right)}^\zeta}d\zeta,\\
		&=\frac{\mathcal{B}_1\Gamma(1-s)}{({-\gamma_t})^{s-1}2\pi j}
		G_{3,3}^{2,2}\left(\begin{array}{c}
			\gamma_t\lambda_1\end{array}
		\Bigg\vert\begin{array}{c}
			(1-m_{1,s},1,0)\\
			(s-1,m_1,0)\\
		\end{array}\right),\label{ap5}
	\end{align}
where $\mathcal{L}_2$ is a specific counter that separates the poles of $\Gamma(-\zeta)$ from the poles of $\Gamma(m_1+\zeta)$. 
Next, by exploiting \cite[Eq. (2.9)]{mathai2009h}, $\mathcal{M}[f_2(\gamma_2),s]$ can be computed as:
	\begin{align}
		\mathcal{M}\left[f_2(\gamma_2),s\right]=\mathcal{A}_2\frac{\Gamma(m_2-1+s)\Gamma(1+m_{2,s}-s)}{{\lambda_2}^s}.\label{ap6}
	\end{align}
	Now, by plugging \eqref{ap5} and \eqref{ap6} into \eqref{ap3}, $P^{CM}_{out}$ can be written as
	\begin{align}
		&\hspace{-2ex}P_{out}=\frac{\gamma_t\mathcal{B}_1\mathcal{A}_2}{4\pi^2}\int_{\mathcal{L}_1}\int_{\mathcal{L}_2}\frac{\Gamma(-\zeta)\Gamma(m_1+\zeta)\Gamma(m_{1,s}-\zeta)}{\Gamma(\zeta)\Gamma(1-\zeta){\left(\gamma_t\lambda_1\right)}^\zeta}\\
		&\hspace{-2ex}\times \Gamma(\zeta+s-1)\Gamma(1-s)\frac{\Gamma(m_2-1+s)\Gamma(1+m_{2,s}-s)}{{\left(-\gamma_t\lambda_2\right)}^s}d\zeta ds,\nonumber
	\end{align}
and finally, exploiting the definition of bivariate Meijer's G-function, the proof is completed. 
	\section{Proof of Theorem \ref{thm-ac-mac}}\label{ac-mac}
	By substituting \eqref{pdf-gs} into \eqref{asc-def}, the ASC can be rewritten as:
	\begin{align}\nonumber
		\bar{\mathcal{C}}^{CM}=&\frac{\mathcal{A}_1\mathcal{A}_2}{\ln 2}\sum_{k=0}^{\infty}\frac{(-\lambda_2)^{k+1}}{k!}\\\nonumber
		&\times G_{3,3}^{2,3}\left(
		\begin{array}{c}
			\frac{-\lambda_1}{\lambda_2}\\
		\end{array}\Bigg\vert
		\begin{array}{c}
			(0,-m_{1,s},k-m_2+1)\\
			(m_1-1,k+m_{2,s},k)\\
		\end{array}\right)\\
	&\times\underset{\mathcal{I}}{\underbrace{{\int_{0}^{\infty}\ln(1+\gamma_s)\gamma_s^kd\gamma_s}}}.\label{as-cm-app}
	\end{align}
Next, by re-expressing the logarithm function in terms of the
Meijer's G-function \cite{prudnikov1990more}, i.e., 
	\begin{align}
		\ln(1+\gamma_s)=
		G_{2,2}^{1,2}\left(\begin{array}{c}
			\gamma_s\end{array}
		\Bigg\vert\begin{array}{c}
			(1,1)\\
			(1,0)\\
		\end{array}\right),\label{log}
	\end{align}
and  using \cite[2.24.2.1]{prudnikov1990more}, the integral $\mathcal{I}$ is computed as:
\begin{align}
\mathcal{I}=\int_{0}^{\infty}
{\gamma}^k_sG_{2,2}^{1,2}\left(\begin{array}{c}
	\gamma_s\end{array}
\Bigg\vert\begin{array}{c}
	(1,1)\\
	(1,0)\\
\end{array}\right)d\gamma_s=\left[\frac{\Gamma\left(-\left(k+1\right)\right)}{\Gamma\left(-k\right)}\right]^2.\label{i}
\end{align}
Now, by inserting \eqref{i} into \eqref{as-cm-app}, the proof is completed. 
\section{Proof of Theorem \ref{thm-ac-dd}}\label{app-ac-dd}
By substituting \eqref{pdf1} into \eqref{j1} and expressing logarithm function in terms of the Meijer's G-function as provided in \eqref{log}, $\mathcal{J}_1$ can be rewritten as:
\begin{align}\nonumber
	\mathcal{J}_1=&\frac{\mathcal{A}_1}{2\lambda_1\ln2}\int_{0}^{\infty}G_{2,2}^{1,2}\left(\begin{array}{c}
		\gamma_n\end{array}
	\Bigg\vert\begin{array}{c}
		(1,1)\\
		(1,0)\\
	\end{array}\right)\\
&\times G_{1,1}^{1,1}\left(
	\begin{array}{c}
		\lambda_1\gamma_n\\
	\end{array}\Bigg\vert
	\begin{array}{c}
		-m_{1,s}\\
		m_1-1\\
	\end{array}\right)d\gamma_n.
\end{align}
Now, with the help of \cite[eqs. (2.25.1.1) and
(8.3.2.21)]{prudnikov1990more}, $\mathcal{J}_1$ is obtained as:
\begin{align}
\mathcal{J}_1=\frac{\mathcal{A}_1}{2\lambda_1\ln2}G_{3,3}^{2,3}\left(
\begin{array}{c}
	\frac{1}{\lambda_1}\\
\end{array}\Bigg\vert
\begin{array}{c}
	(1,1,1-m_{1})\\
	(1,m_{1,s},0)\\
\end{array}\right).
\end{align}

By inserting \eqref{pdf1} and \eqref{cdf1} into \eqref{j2}, considering \eqref{log}, and exploiting the definition of Meijer's G-function, $\mathcal{J}_2$ can be determined as:
\begin{align}\nonumber
	&\mathcal{J}_2=\frac{\mathcal{A}_1\mathcal{B}_2}{2\ln2}\int_{0}^{\infty}G_{2,2}^{1,2}\left(\begin{array}{c}
		\lambda_2\gamma_n\end{array}
	\Bigg\vert\begin{array}{c}
		(1-m_{2,s},1)\\
		(m_2,0)\\
	\end{array}\right)\\
	&\times G_{2,2}^{1,2}\left(\begin{array}{c}
		\gamma_n\end{array}
	\Bigg\vert\begin{array}{c}
		(1,1)\\
		(1,0)\\
	\end{array}\right)G_{1,1}^{1,1}\left(
	\begin{array}{c}
		\lambda_1\gamma_n\\
	\end{array}\Bigg\vert
	\begin{array}{c}
		-m_{1,s}\\
		m_1-1\\
	\end{array}\right)d\gamma_n,\\\nonumber
	&=\frac{\mathcal{A}_1\mathcal{B}_2}{2\ln2}\int_{\mathcal{L}_1}\frac{\Gamma(m_2+s)\Gamma(m_{2,s}-s)\Gamma(-s)}{\lambda^s_2\Gamma(1-s)}\\
	&\hspace{-0.3cm}\footnotesize\times\underset{\mathcal{K}}{\underbrace{\int_{0}^{\infty}\gamma^{-s}_nG_{2,2}^{1,2}\left(\begin{array}{c}
				\gamma_n\end{array}
			\Bigg\vert\begin{array}{c}
				(1,1)\\
				(1,0)\\
			\end{array}\right)G_{1,1}^{1,1}\left(
			\begin{array}{c}
				\lambda_1\gamma_n\\
			\end{array}\Bigg\vert
			\begin{array}{c}
				-m_{1,s}\\
				m_1-1\\
			\end{array}\right)d\gamma_nds}},\label{j21}
\end{align}
where $\mathcal{L}_1$ is a certain contour separating the poles of $\Gamma(m_2+s)$ from the poles of $\Gamma(-s)$. Next, by utilizing the Mellin transform for the product of two Meijer's G-functions
\cite[eq. (2.25.1.1)]{prudnikov1990more}, the inner integral $\mathcal{K}$ can be computed as follows:
\begin{align}
	\mathcal{K}={\lambda}^{s-1}_1G_{3,3}^{2,3}\left(
	\begin{array}{c}
		\frac{1}{\lambda_1}\\
	\end{array}\Bigg\vert
	\begin{array}{c}
		(1,1,s-1-m_1)\\
		(1,m_{1,s}+s,0)\\
	\end{array}\right),\label{k}
\end{align}
subsequently, by applying the definition of univariate Meijer's G-function to \eqref{k}, then plugging the obtained
result into \eqref{j21} and performing the change of variables
$s=-s$ and $\zeta=-\zeta$, we have:
\begin{align}\nonumber
	\mathcal{J}_2=&-\frac{\mathcal{A}_1\mathcal{B}_2}{8\lambda_1{\pi^2}\ln2}\int_{\mathcal{L}_1}\int_{\mathcal{L}_2}\frac{\Gamma(m_{1,s}-s-\zeta)\Gamma(s)}{\Gamma(1+s)\Gamma(1+\zeta)}\\\nonumber
	&\times\Gamma(m_1+s+\zeta)\Gamma(m_2-s)\Gamma(m_{2,s}+s)\\
	&\times\Gamma^2(\zeta)\Gamma(1-\zeta)\left(\frac{1}{\lambda_1}\right)^\zeta\left(\frac{\lambda_2}{\lambda_1}\right)^sd\zeta ds,
\end{align}
where $\mathcal{L}_2$ is another contour. Consequently, recognizing the  definition of bivariate Meijer’s G-function \cite{gupta1969integrals}, $\mathcal{J}_2$ is derived as:
\begin{align}\nonumber
	&\mathcal{J}_2=\frac{\mathcal{A}_1\mathcal{B}_2}{2\lambda_1\ln 2}\\
	&\times G_{1,1:2,2:2,2}^{1,1:1,2:1,2}\left(\begin{array}{c}
		\frac{\lambda_2}{\lambda_1},\frac{1}{\lambda_1}\end{array}
	\Bigg\vert\begin{array}{c}
		(m_1)\\
		(m_{1,s})\\
	\end{array}\Bigg\vert
	\begin{array}{c}
		(1-m_{2,s},1)\\
		(m_2,0)\\
	\end{array}\Bigg\vert
	\begin{array}{c}
		(1,1)\\
		(1,0)\\
	\end{array}\right).
\end{align}
Similarly, following the same methodology, $\mathcal{J}_3$ and $\mathcal{J}_4$ are respectively determined as:
\begin{align}
\mathcal{J}_3=\frac{\mathcal{A}_2}{2\lambda_2\ln2}G_{3,3}^{2,3}\left(
	\begin{array}{c}
		\frac{1}{\lambda_2}\\
	\end{array}\Bigg\vert
	\begin{array}{c}
		(1,1,1-m_{2})\\
		(1,m_{2,s},0)\\
	\end{array}\right),\label{j3}
\end{align}
\begin{align}\nonumber
	&\mathcal{J}_4=\frac{\mathcal{B}_1\mathcal{A}_2}{2\lambda_2\ln 2}\\
	&\times G_{1,1:2,2:2,2}^{1,1:1,2:1,2}\left(\begin{array}{c}
		\frac{\lambda_1}{\lambda_2},\frac{1}{\lambda_2}\end{array}
	\Bigg\vert\begin{array}{c}
		(m_2)\\
		(m_{2,s})\\
	\end{array}\Bigg\vert
	\begin{array}{c}
		(1-m_{1,s},1)\\
		(m_1,0)\\
	\end{array}\Bigg\vert
	\begin{array}{c}
		(1,1)\\
		(1,0)\\
	\end{array}\right).
\end{align}
Finally, by substituting $\mathcal {J}_l$ for $l\in\{1,2,3,4\}$ into \eqref{j1-j4}, the proof is accomplished.
\bibliographystyle{IEEEtran}
\bibliography{sample.bib}

\end{document}